\documentclass[twocolumn,longbibliography]{revtex4-1}
\usepackage{amsmath} 
\usepackage{amsfonts} 
\usepackage{amssymb} 
\usepackage{graphicx}
\usepackage{psfrag}
\usepackage{enumitem}

\newcommand{\id}{\mathbf{1}}

\newcommand{\suchthat}{\,|\,}
\newcommand{\set}[2]{ \{#1\suchthat #2\} }

\newcommand{\sset}[1]{ \{#1\} }

\newcommand{\chan}[1]{\mathcal{#1}}
\newcommand{\cchan}[1]{\mathbf {#1}}
\newcommand{\pchan}[2]{\sset{{#1}^{\frac 1 2}\,{#2}}}

\newcommand{\hilb}{\mathcal H}
\newcommand{\sub}{\subseteq}

\newcommand{\syndx}{\mathfrak S}
\newcommand{\synd}[1]{\mathfrak S{#1}}
\newcommand{\corr}[1]{\mathfrak C_{#1}}
\newcommand{\ceff}[1]{{#1}_\mathrm {eff}}
\newcommand{\echan}[1]{\ceff{\chan {#1}}}
\newcommand{\reduce}[1]{\overline {#1}}
\newcommand{\rchan}[1]{\chan{\reduce #1}}
\newcommand{\rechan}[1]{\reduce {\echan #1}}
\newcommand{\fail}[1]{\mathrm {fail}\,(#1)}

\newcommand{\supp}[1]{\mathrm {Supp}\,#1}
\newcommand{\err}[1]{\mathrm {Err}\,#1}
\newcommand{\errx}{\mathrm {Err}}

\newcommand{\stab}{\mathcal S}
\newcommand{\cent}[1]{\mathcal Z(#1)}
\newcommand{\gauge}{\mathcal G}
\newcommand{\logical}{\mathcal L}
\newcommand{\pauli}{\mathcal P}

\newcommand{\qed}{\hspace*{\fill}$\blacksquare$}
 
 \newtheorem{thm}{Theorem}
 \newtheorem{lem}[thm]{Lemma}
 \newtheorem{cor}[thm]{Corolary}
 \newtheorem{defn}[thm]{Definition}
 \newtheorem{prop}[thm]{Proposition}
 \newenvironment{proof}{\noindent \emph{Proof.}}{\qed}
 
 \newenvironment{sketchproof}{\noindent \emph{Sketch of proof.}}{\qed}

\begin{document}

\title{Single-shot fault-tolerant quantum error correction}

\author{H\'ector Bomb\'in}
\affiliation{Perimeter Institute for Theoretical Physics, 31 Caroline St. N., Waterloo, Ontario N2L 2Y5, Canada}
\affiliation{Deparment of Mathematical Sciences, University of Copenhagen, Universitetsparken 5, DK-2100 Copenhagen \O}

\begin{abstract}
Conventional quantum error correcting codes require multiple rounds of measurements to detect errors with enough confidence in fault-tolerant scenarios.
Here I show that for suitable topological codes a single round of local measurements is enough.
This feature is generic and is related to self-correction and confinement phenomena in the corresponding quantum Hamiltonian model.
3D gauge color codes exhibit this single-shot feature, which applies also to initialization and gauge-fixing.
Assuming the time for efficient classical computations negligible, this yields a topological fault-tolerant quantum computing scheme where all elementary logical operations can be performed in constant time.\end{abstract}


\maketitle

\section{Introduction}

The development of efficient fault-tolerant quantum computing techniques is essential: they provide the means to deal with the decoherence and control imprecisions that are intrinsic to quantum systems, see \emph{e.g.}~\cite{lidar:2013:quantum}.
Only the presence of such error sources prevents extending the ability to control small quantum systems for a limited time to that of performing arbitrarily long, precise and large quantum computations.

A key element of fault-tolerant quantum computing is error correction, which is intended to detect and eliminate errors in the system.
But error detection is itself a noisy process, and the attempt to eliminate wrongly diagnosed errors may end up introducing many more.
This problem can be addressed~\cite{shor:1996:ftqc} by performing multiple times the measurements from which errors are to be inferred~\footnote{Alternatively this effort can be transferred to the preparation of highly entangled states \cite{steane:1997:active,knill:2005:scalable}.}.
As a drawback, longer times and more operations lead themselves to more accumulation of errors.
Fortunately the accumulation is not catastrophic and fault-tolerance can be achieved~\cite{aharonov:1997:ftqc,kitaev:1997:qec,knill:1998:resilient}.

This paper deals with an alternative solution to the problem of wrong diagnoses. 
It shows how certain error correction strategies are robust against imperfections in error detection: for them a single round of \emph{local} measurements suffices.

\subsection{Locality}

\begin{figure}
  \centering
  \includegraphics[width=7cm]{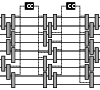}
  \caption{
   A quantum-local process involves a finite number of local operation rounds, each with a classical output that is globally processed to provide an input for the next round.
   }
\label{fig:qlocal}
\end{figure}

Locality plays a crucial role in fault-tolerant quantum computing techniques.
This is rooted in the fact that physical interactions have a local nature~\cite{preskill:2012:sufficient}.
Accordingly, quantum information is encoded in non-local degrees of freedom that are not directly accessible to the environment.
Local degrees of freedom can then absorb the damage caused by errors as the computation proceeds.

In order to compute it must be possible to initialize, transform and measure the encoded quantum information.
Moreover, there must be a way to extract information about the errors that afflict the system.
All such operations on encoded information should be designed so that they preserve the local structure of noise.

A natural way to achieve this is to make use of local operations, in the following sense.
Of interest here are systems composed by many physical qubits. 
A local operation is one that can be implemented with a quantum circuit of bounded depth, possibly including the use of ancilla qubits.
The elementary operations of these circuits can involve only a bounded number of qubits but are otherwise arbitrary.
In addition, locality can have a geometric meaning if qubits, including the ancillas, are arranged on a lattice and elementary operations involve only neighboring qubits.

Local operations are great in that they preserve the local structure of noise, \emph{e.g.} an error afflicting a single qubit will be mapped to an error afflicting a few qubits, where the exact number will depend on the depth of the circuit and the number of qubits involved in each elementary operation.
Unfortunately local operations also suffer from limitations, as exemplified by no-go theorems on the computational power of certain classes of local operations~\cite{eastin:2009:restrictions,bravyi:2013:classification,pastawski:2015:fault,beverland:2014:protected}.

A way out of these restrictions is to consider a larger set of operations in which classical information can be non-locally processed.
In particular, consider operations that can be decomposed in a bounded number of local processes, so that the classical output of each is processed globally to produce the classical input of the next, see Fig.~\ref{fig:qlocal}.
Such operations, which will be termed \emph{quantum-local}, do not suffer from the computational limitations of purely local operations mentioned above~\cite{bombin:2013:gauge}.
On the other hand, quantum-local operations do not automatically preserve the local structure of noise.
However, carefully chosen quantum-local operations can have this property as long as classical information processing is noiseless.
This is a central message of the present work.

\subsection{Single-shot error correction}

In some quantum error-correcting schemes the measurements needed for error detection are local, and so are the operations involved in correcting a given \emph{error syndrome}.
This is the case, for example, of the toric code~\cite{kitaev:1997:qec}.
But processing the error syndrome in order to choose the right correcting operator, and controlling the application of this operator, requires global classical communication and computation.
Error correction is thus in such cases a quantum-local operation.

More specifically, noiseless error correction is quantum-local.
But as soon as the measurements are noisy, the error syndrome becomes unreliable and large errors might be introduced through the correction operation.
To improve accuracy measurements can be repeated as much as needed, but then in a fault-tolerant setting quantum-locality is lost because higher precision requires more measurements~\cite{dennis:2002:tqm}:
the number of measurement rounds is unbounded and the fault-tolerant error correction operation as a whole is not quantum-local, even though each round of measurements is by itself a quantum-local operation.

As discussed below, in some schemes localized syndrome measurement errors only give rise to localized errors in the correction stage.
Then a single round of measurements is enough, and fault-tolerant error correction is quantum-local.
The error syndrome will be noisy and after error correction there will be residual local noise, but this is not a problem.
The goal is not, and cannot be, to remove all entropy in the system, but just to keep it low so that errors do not accumulate and damage the computation.

This approach to error correction will be termed single-shot fault-tolerant quantum error correction, or \emph{single-shot error correction} for short.

\subsection{Self-correction}

All the quantum error correction schemes considered in this paper are topological~\cite{dennis:2002:tqm}.
In particular, they are such that (i) physical qubits are arranged in a lattice of variable size, (ii) error correction is quantum-local with respect to the lattice and (iii) logical operators (those acting on encoded qubits) have a  support comparable to the lattice size.
Codes with these properties are expected to have a noise threshold: for local noise below the threshold value, noiseless error correction is successful with probability almost one in the thermodynamic limit.
In particular, the effective error rate is exponentially small in the lattice size.

In conventional topological codes quantum information is encoded in a subspace of the Hilbert space of the many-qubit system, the code subspace.
Since error detection measurements are local, this subspace has a local description and one can write a local gapped Hamiltonian such that its ground state is the code subspace.
The resulting condensed matter system is said to be topologically ordered~\cite{wen:1989:to}: there exist a ground state degeneracy with a non-local origin.

Some of these topologically ordered phases are special in that they survive at finite temperatures.
This means that the topological degeneracy of the ground state is not affected by thermal noise, so that such systems would provide a self-correcting quantum memory: a quantum memory that does not require active error correction~\cite{dennis:2002:tqm, alicki:2010:thermal}.
For a recent review on this topic, see~\cite{brown:2014:qm}.

The results below show that there exists a connection between self-correction and single-shot error correction.
Namely, all the known topological codes that yield self-correcting systems turn out to exhibit single-shot error correction too.
The connection boils down to the fact that in both cases the connectivity of excitations/error syndrome gives rise to their confinement, i.e. large connected sets of excitations are highly unlikely to be observed.

\subsection{Constant time overhead}

Unfortunately known self-correcting phases require at least four spatial dimensions.
Thus one does not expect to find a 3D topological encoding scheme allowing single-shot error correction.
But it turns out that this is possible by considering \emph{subsystem} topological codes rather than conventional ones.

Subsystem codes are those in which the code space contains both logical and gauge qubits. 
The latter are just qubits that can be in any unknown state, and in a topological subsystem code they might include local degrees of freedom.
This local degrees of freedom can be measured when recovering the error syndrome, providing extra redundancy that can result in the desired single-shot error correction.

Gauge color codes are a remarkable class of topological subsystem codes~\cite{bombin:2013:gauge}.
In particular, 3D gauge color codes allow the quantum-local implementation of a universal set of gauges via gauge fixing~\cite{paetznick:2013:universal}, a technique that uses error correction to switch between different encoding schemes.
As shown below, in 3D gauge color codes all the elementary logical operations (initialization, gates, measurements) remain quantum-local even in the fault-tolerant scenario.
Thus, they provide a fault-tolerant quantum computing scheme with a constant time overhead.
That is, disregarding classical computations, which in any case can be performed efficiently.

Single-shot error correction in gauge color codes can also be interpreted as a form of confinement, but in this case the confined objects are point-like instead of extended. In fact, the confinement mechanism is related to that appearing in lattice gauge theories, where certain point-like field sources cannot be isolated because they are connected to other sources through flux tubes with an energy density. 
The key difference is that confinement in the single-shot context does not follow from energetic considerations, which are absent, but instead from analogous probabilistic considerations.

\section{Confinement}\label{sec:confinement}

This sections intends to give the essential picture of (i) the connection between self-correction and single-shot error correction, and (ii) the mechanism behind single-shot error correction in $D=3$ spatial dimensions.
The unifying concept will be confinement, either of excitations in a self-correcting memory or of syndromes in error correction.

\subsection{Ising model and repetition code}\label{sec:Ising}

Maybe the simplest conceivable (classical) self-correcting memory is the Ising model.
It displays a critical temperature below which a classical bit can be stored reliably for a time that grows exponentially with the system size.
In particular this holds for spatial dimensions $D\geq 2$.
This is due to the fact that excitations have dimension $D-1$, and they are only confined when they are extended objects.
\emph{E.g.} for $D=2$ excitations take the form of loops.
Each loop of length $l$ carries a Boltzmann factor of the form $e^{-\beta l}$.
Below the critical temperature this factor is stronger that entropy fluctuations and loops are confined, i.e. large loops are highly unlikely to be observed.
The suppression of large loops protects the classical bit because loops of a length comparable to the system size are required to flip its value.

It is possible to construct a quantum code based on the Ising model.
It will correct only bit-flip errors, but this is enough to illustrate the role that confinement can play in quantum error correction.

The construction is as follows.
Physical qubits are placed on the faces of a 2D square lattice with periodic boundary conditions, see Fig.~\ref{fig:Ising}.
As usual $X_i, Z_i$ denote the $X, Z$ Pauli operators on the $i$-th qubit.
Between any pair of adjacent faces $i,j$ there is an edge $e$.
For each such edge $e$ there is a \emph{check operator}
\begin{equation}
Z_e:=Z_iZ_j
\end{equation}
that will be measured to recover the error syndrome.
Encoded states are those for which $Z_e=1$, \emph{i.e.} eigenstates of $Z_e$ with eigenvalue 1.
These are superpositions of the states $|0\rangle^{\otimes n}$ and $|1\rangle^{\otimes n}$, where $n$ is the number of physical qubits.
Thus, this is just a repetition code.

An error syndrome can be identified with the collection $l$ of edges such that the measurements yield $Z_e=-1$.
If a number of qubits are flipped in an encoded state, the resulting syndrome $l$ will be the boundary of the area corresponding to these qubits, see Fig.~\ref{fig:Ising}.
Not any collection of edges $l$ can be a syndrome: $l$ has to be the boundary of a region.
In particular, an even number of edges of $l$ must meet at every vertex of the square lattice.
For every $l$ there are two possible sets of qubits that could be flipped to produce it.
When correcting the syndrome $l$, one must flip the qubits of one of these sets, denoted $l_+$.
The complementary set of qubits is denoted $l_-$.
An error that flips the set of qubits $l_-$ is undesirable: correcting it produces a logical error.

For local noise of low enough intensity flipped qubits will form small clusters.
The resulting syndrome $l$ will thus be composed of small disconnected loops.
In particular, if each qubit is flipped independently with probability $\lambda$, a given set of edges $s$ can be a subset of the syndrome $l$ with probability
\begin{equation}\label{eq:prob_s}
p(s)\leq (2\lambda)^{|s|/4},
\end{equation}
since for every edge in $s$ there are two qubits that could have been flipped, and each flip affects at most 4 edges. 
Thus loops are exponentially suppressed as in the Ising model, and for $\lambda$ under a critical value this effect dominates over entropic fluctuations and syndrome loops are confined.

\begin{figure}
  \centering
  \includegraphics[width=8.6cm]{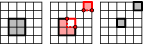}
  \caption{
  A quantum error-correcting code based on the 2D Ising model.
  Qubits sit at faces and check operators at edges.
  Three stages of error correction are depicted.
 (Left)  When the qubits in the shadowed area are flipped the check operators on the boundary detect the change.
  (Middle) Noisy measurements of the syndrome are performed. 
  Measurements fail at some edges (red) providing a pseudo-syndrome that is not closed and thus has endpoints (red circles).
  The failed measurements are estimated to correspond to a minimal set of edges with the same endpoints (dotted red), so that the effectively recovered syndrome is the boundary of the shaded area.
  (Right) After error correction is performed, the new syndrome corresponds to the set of edges with effectively wrong measurement outcome.
  }
\label{fig:Ising}
\end{figure}

\subsection{Noisy error correction}\label{sec:noisy}

Noiseless error correction aims to produce logical errors with a probability that is as small as possible.
For this it suffices to choose $l_+$ so that it coincides with the most likely error for a given syndrome loop $l$.
In noisy error correction, however, the form of the unavoidable residual noise is also important.
For the present case, it should be such that syndrome loops are confined if measurement errors do not happen too often.

Suppose that check operators $Z_e$ are measured and those at edges $e$ belonging to a certain set $w$ give the wrong eigenvalue.
Instead of the correct syndrome $l$ the recovered pseudo-syndrome is $l+w$, with $+$ the symmetric difference of sets.
This needs not be a proper syndrome since it might not be closed, and therefore it needs to be corrected. 
There exists some set of edges $w_0$ with minimal cardinality such that $l+w+w_0$ is closed.
Clearly $w_0$ only depends on $w$, because $l+w+w_0$ is closed if and only if $w+w_0$ is.
Since $w+w=\emptyset$ is closed it follows that
\begin{equation}\label{eq:mini}
|w_0| \leq |w|.
\end{equation}
Estimating that $w_0$ is the set of wrong measurements yields an effective set of wrong measurements $w'$ and an effective syndrome $l'$
\begin{equation}\label{eq:eff_synd}
w':=w+w_0, \qquad l' := l+w'.
\end{equation}
The situation is depicted in Fig.~\ref{fig:Ising}.
In addition $l'$ has to satisfy global restrictions: it has to be a boundary, not just closed.
Here for simplicity it will be assumed that $l'$ is indeed a boundary, but the problem is discussed in section~\ref{sec:global_constraints}.

To explore the effects of noisy measurements, assume that the original error was $l_+$, so that noiseless error correction would not introduce a logical error.
In order to correct the error syndrome $l'$ the qubits in $l'_+$ are flipped.
Therefore the net effect of noise and correction is to flip the set of physical qubits
\begin{equation}\label{eq:new_synd}
l_++l'_+=w'_\pm,
\end{equation}
where the $\pm$ indicates the two different possibilities, see Fig.~\ref{fig:Ising}.
The minus sign corresponds to a logical error.

The set of effective wrong measurements $w'$ is also the syndrome of the residual error $w'_\pm$.
Does local noise in measurements give rise to confined syndrome loops? 
Suppose that each measurement fails with probability $\eta$.
According to \eqref{eq:mini} more than half of the elements of $w'$ are also in $w$.
Since there are $2^{|w'|}$ subsets of $w'$, the probability for $w'$ to be the syndrome is bounded by $(2 \eta^{1/2})^{|w'|}$.
With a bit more of care one can show that for small enough $\eta$ the probability that a set of edges $r$ is a subset of $w'$ is
\begin{equation}\label{eq:prob_r}
p(r)\leq \upsilon^{|r|},
\end{equation}
for some $\upsilon$ that goes to zero when $\eta$ does, see section~\ref{sec:self-correction}.
Thus indeed residual syndrome loops are confined.

An interesting way to choose the correcting set $l_+$ for each $l$ is to correct each connected component of $l$ separately~\cite{bombin:2013:self}, so that localized loops are corrected by flipping a localized set of qubits.
This choice makes it easy to understand that logical errors are unlikely when both the original error syndrome $l$ and the measurement error $w'$ are confined in the sense discussed.
Indeed, the edges in the union of $l$ and $l'$ will form clusters, and \eqref{eq:new_synd} can be applied separately to each of them.
If these clusters are small compared to the lattice, the set $l_++l'_+$ is localized and has to be $w'_+$.
Logical errors due to noisy measurements will only happen when the clusters are large compared to the system size, and the probability for such events decreases exponentially with the lattice size as long as confinement is strong enough.

The 2D Ising code enables single-shot error correction for bit-flip errors when the probability of measurement errors is low enough.
The procedure is quantum-local because it consists of (i) measuring the check operators, a local operation, (ii) computing $l'_+$, a classical operation, and (iii) flipping the qubits in $l'_+$, a local operation.
Error correction is successful almost with certainty in large systems (assuming of course that the original noise to be corrected is weak enough) and the residual noise is also weak in the sense that its syndrome loops exhibit confinement.
This in turn is enough for the system to be able to absorb new errors, which is the goal of fault-tolerant error correction.

\subsection{Spatial dimension}

For $D=1$ spatial dimensions the Ising model fails to have a finite critical temperature.
The excitations are punctual and can move freely without any energy cost: they are unconfined. 
This behavior is mimicked by the corresponding $1D$ repetition code.
Local noise will still give rise to confined syndrome points, in the sense that each connected chain of errors produces two endpoints that are unlikely to be far apart.
But each wrong check operator measurement will give rise to an isolated syndrome point, and thus there is no confinement.

The situation is somewhat similar when quantum self-correcting systems are considered.
However, for known systems the spatial dimension needs to be at least four, instead of two.
Again the mechanism for confinement involves that all excitations are extended objects.
This translates nicely when moving to active error correction.
In particular, again one finds that confinement gives rise to single-shot error correction.

\subsection{Previous confinement mechanisms}

Confinement is not only important for single-shot error correction.
Rather, it is possibly a necessary element of any fault-tolerant error correction scheme involving local codes. 
This is at least the case for transversal measurements and for repeated syndrome extraction, as discussed next.

To illustrate this point it suffices to consider the 1D repetition code.
Fault-tolerant error correction was argued above not be achievable by quantum-local means for this code.
Nevertheless, it can be achieved if an unbounded number of rounds of syndrome extraction are allowed~\cite{dennis:2002:tqm}, as recently demonstrated experimentally~\cite{kelly:2015:state}.
In essence, time provides an extra dimension that enables the confinement of errors: the point-like syndrome measurement errors occurring at different rounds can be put together to form `world-lines', \emph{i.e.} the required extended objects. 
The complete picture has to include the errors that will affect the qubits as the rounds proceed, but the end result is that the word-lines can indeed be confined exactly as in the 2D code, and fault-tolerance is achieved~\cite{dennis:2002:tqm, gottesman:2013:overhead}.

There is at least another context where the 1D repetition code displays confinement: the measurement of the logical qubit in the computational basis.
This measurement can be carried out fault-tolerantly with a quantum-local operation.
Namely, it suffices to (i) measure each physical qubit in the computational basis, and (ii) classically process the result to decode the eigenvalue of the logical $Z$. 
This last step is not different from the classical decoding of a repetition code.
Why is this process fault-tolerant?
The straightforward answer is that measurement errors can be regarded as bit-flip errors prior to the measurements, mapping the problem back to ideal error correction.
An alternative answer is that the structure of the measurements gives rise to confinement.
The key is that the individual measurement outcomes correspond to $Z_i$ eigenvalues, not the check operators $Z_iZ_j$.
This means that wrong measurements can be regarded as strings, not the endpoints of strings, so that long error strings pay a probabilistic penalty that yields the desired confinement.
The wrong check operator outputs are the endpoints of the strings formed by wrong measurements, \emph{e.g.} if measurements fail from qubits $i$ to $j$ and nowhere else then $Z_{i-1}Z_i$ and $Z_jZ_{j+1}$ are assigned a wrong eigenvalue.
The endpoints of confined strings are also confined, in the following sense: it is unlikely to have an isolated wrong check operator output, because that requires a long connected string of measurement errors.

\subsection{Charge confinement in 3D codes}

\begin{figure}
  \centering
  \includegraphics[width=8.6cm]{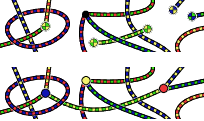}
  \caption{
  Two stages in the extraction of the error syndrome.
  (Top) The measurements of the gauge generators provides a noisy gauge syndrome that does not satisfy color flux conservation at the points marked with a two-colored circle.
  The branching point marked in black satisfies flux conservation.
  (Bottom) The gauge syndrome is corrected by adding suitable edges.
  The error syndrome is the set of branching points of the net of fluxes, marked here as colored circles.
  }
\label{fig:branching}
\end{figure}

Surprisingly, with subsystem codes it is possible to reproduce in $D=3$ spatial dimensions the phenomena of confinement.
More precisely, it will be proved below that this is possible at the level of noisy error correction.

The codes that turn out to show this behavior are $3D$ gauge color codes~\cite{bombin:2013:gauge}.
As in the 1D repetition code, errors can be visualized as strings and syndromes as the endpoints of these strings, all living in a 3D lattice.
Direct extraction of the error syndrome provides no confinement but, instead, it is possible to perform a collection of local measurements involving the gauge degrees of freedom.
The resulting \emph{gauge syndrome} is composed of extended objects similar to the closed strings in the 2D repetition code, and noisy measurements will display confinement.
As described below, the wrong syndrome outputs, which are point-like, are geometrically part of these extended objects.
Moreover, the confinement of the extended objects yields the confinement of the point-like wrong syndromes.

As discussed at the end of the previous subsection, such an `inherited' confinement of point-like wrong syndrome outputs occurs also for transversal measurements in the 1D repetition code.
The two mechanisms are somewhat similar.
When the logical $Z$ operator is to be measured in the 1D code, $Z_i$ operators become `gauge' operators: phase errors do not affect the outcome of the measurement in any way~\footnote{The discussion here is closely connected to the notion of classical (topological) gauge codes~\cite{bombin:2013:structure}.}.
In particular, one can measure the $Z_i$, with the consequences described above.
Once again, wrong syndrome outputs ($Z_iZ_j$ eigenvalues) are the endpoints of extended objects (strings) that consist of wrong measurement outputs (of $Z_i$ gauge operators).
The case of 3D gauge color codes is analogous: the extended objects are string-nets, and the point-like objects are the branching points of these string-nets.

Details are as follows.
The 3D lattice has vertices with four different colors: red, green blue and yellow (r, g, b, y).
These are connected through edges that have two colors, namely those complementary to the colors of the two vertices linked by the edge, which have to be different.
E.g. an rg-edge connects a b-vertex and a y-vertex.
The error syndrome consists of a collection of vertices.
The gauge syndrome consists of a collection of edges that has to be \emph{closed} in the following sense.
The subset of edges on the gauge syndrome with a given color on them has to be closed, \emph{i.e.} composed of loops.
E.g. the set of edges carrying red color (namely rg-, rb- and ry-edges) in the gauge syndrome, must form closed loops, \emph{i.e.} an even number of them must meet at every vertex.
In other words, color flux is conserved.

The error syndrome can be recovered from the gauge syndrome as follows:
a vertex is in the error syndrome if and only if an odd number of gauge syndrome edges of any color is incident on it.
\emph{E.g.} if an odd number of gauge syndrome rg-edges meet at a b-vertex, then it is an error syndrome vertex (and then necessarily and odd number of ry-edges, and of yg-edges, must meet at the vertex to preserve the color flux).
As depicted in Fig.~\ref{fig:branching} (bottom), error syndrome vertices are the branching points of the gauge syndrome net of fluxes.

The conservation of color flux is analogous to the situation in the 2D repetition code, where the syndrome lines had to be closed.
Thus, as in that case, one has to repair the noisy gauge syndrome to ensure that color flux is indeed conserved, as illustrated in Fig.~\ref{fig:branching}.
In the 2D case the failure of a faulty syndrome to be closed amounts to the presence of endpoints, \emph{e.g.} the red circles of Fig.~\ref{fig:Ising} (middle).
For the 3D gauge color codes there are also endpoints, but now they come in different types due to the coloring, see Fig.~\ref{fig:branching} (top).
One could say that in the 2D case the flux comes in a single color, whereas in 3D there are several colors that separately have to form closed lines.

Since gauge degrees of freedom are not protected in any way, the gauge syndrome is \emph{a priori} random, except for the constraints imposed by the relationship to the error syndrome.
Thus the gauge syndrome loops are not confined at all, but this is immaterial.
As in the 2D repetition code, there will be a set $w$ of edges that gave the wrong measurement, and a minimal set $w_0$ such that $w'=w+w_0$ is closed.
Unlike the original set of wrong measurements $w$, the effective set of wrong measurements $w'$ satisfies color flux conservation.
It is $w'$ that displays confinement, and this is what matters.
Let the correct error syndrome be $v$, a collection of vertices.
Instead the recovered syndrome is $v+v'$, with $v'$ the branching points of $w'$.
\emph{I.e.} $v'$ is the set of wrong syndrome outcomes.
If $w'$ is confined it will typically consists of small clusters, each with a collection of branching points that are also thus clustered, see Fig.~\ref{fig:confinement}.
This is how the point-like wrong syndrome outcomes $v'$ inherit confinement.

It is possible to attach a charge to the branching points~\cite{bombin:2007:branyons}.
Under this perspective each connected component of $w'$ gives rise to a collection of charges with neutral total charge.
In other words, the charge is confined.
In fact, one can regard the charges $v'$ as sources of a field with flux tubes given by the elements of $w'$.
In this picture the connection between the gauge syndrome and the error syndrome is just the gauss law.

The error correction just described must be performed twice, one for $X$ errors and the other for $Z$ errors.
The two cannot be unified in the sense that the gauge operators to be measured do not commute.
Also for this reason it is not possible to upgrade the gauge syndrome to an error syndrome.
This would work only for either $X$ or $Z$ errors, and confinement would be lost for the other type of errors.

\begin{figure}
  \centering
  \includegraphics[width=5cm]{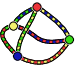}
  \caption{
  The set of edges with an effectively wrong measurement constitute a gauge syndrome. 
  If noise in measurements is below a threshold, this gauge syndrome is composed of small clusters. 
  Each cluster contributes branching points that are thus close to each other.
  Moreover, they have overall neutral charge, \emph{i.e.} there exists a local error with such a syndrome.
  This is the origin of confinement.
  }
\label{fig:confinement}
\end{figure}

\section{Models}\label{sec:models}

This section deals with the modeling of fault-tolerant quantum error correction.
The focus will be on simplicity, but keeping the models interesting enough so that the results are compelling.

\subsection{Framework}

Strictly speaking fault-tolerant quantum computing requires studying the computation process as a whole.
In order to isolate the error correction stage, however, here the focus will be on noisy channels, in particular classes (sets) of channels where noise satisfies certain criteria.
Given classes of channels $\cchan A$ and $\cchan B$ their composition is
\begin{equation}
\cchan A\circ\cchan B:=\set{\chan A\circ \chan B}{\chan A\in \cchan A, \chan B\in \cchan B},
\end{equation}
and similarly for the composition $\cchan A\circ \chan E$ or $\chan E\circ \cchan A$ with some channel or operation $\chan E$.

The noise that accumulates in a computation at a given time will be represented by a class of channels $\cchan N_{\tau,\epsilon}$.
The specific class will depend on the error correction scheme of interest, but the meaning of the (positive) parameters $\tau$ and $\epsilon$ will always be the same:
\begin{itemize}
\item
$\tau$ characterizes the accumulation of errors, which is partially reversible but becomes catastrophic if $\tau$ surpasses a critical value, and
\item
$\epsilon$ characterizes logical noise, which is not correctable in any way.
\end{itemize}
The reasons to consider two separate kinds of noise will become clear shortly.
Noisy error correction will be represented by a class of channels $\cchan R_\eta$ (where the R stands for recovery).
Again, the specific class will depend on the scheme of interest, but the positive parameter $\eta$ will always indicate how noisy the recovery operation is.
In the case of single-shot error correction, the recovery operation is quantum-local.

The purpose of the recovery operation is (i) to reduce the value of $\tau$ as much as possible, so that the system can absorb more errors, and (ii) to increase the value of $\epsilon$ as little as possible.
To formalize this, let $\chan C$ denote the map that projects onto the code subspace.
The goal is to find conditions under which
\begin{equation}\label{eq:goal}
\cchan R_{\eta}\circ \cchan N_{\tau,\epsilon}\circ \chan C\subseteq \cchan N_{\tau',\epsilon+\delta}\circ \chan C,
\end{equation}
in such a way that (i) the residual noise $\tau'$ can be made as small as desired solely by reducing the noise $\eta$ of the recovery operation and (ii) the increment $\delta$ of the logical noise can be made as small as desired solely by increasing the system size.

In practice $\tau$ will be characterized in terms of error syndromes, so that it nicely complements the information given by $\epsilon$.
As for $\eta$, it will eventually be directly related to the likeliness of syndrome extraction errors.
The whole approach to be used follows the line of reasoning of section~\ref{sec:noisy}.
The parameter $\tau$ is intended to play the role of a `temperature' that describes how much the error syndrome $l$ is confined.
Recall in this regard the parallelism of section ~\ref{sec:Ising} between excitations and error syndrome: confinement emerged for low temperature and low noise levels, respectively.
Similarly $\eta$ dictates how much the wrong error syndrome $w'$ is confined, see~\eqref{eq:eff_synd}.
Thus \eqref{eq:goal} states that error correction aims to `refrigerate' the system while disturbing encoded qubits as little as possible.
The noise in the recovery operation puts a lower limit on the residual temperature $\tau'$.

\subsection{Stabilizer codes}\label{sec:stabilizer}

In operator quantum error correction~\cite{kribs:2005:unified}, quantum information is stored in a subsystem $B$ of a \emph{code subspace} 
\begin{equation}
C= A\otimes B
\end{equation}
of the Hilbert space $\hilb$ representing the noisy system.
Thus $A$ represents \emph{gauge degrees of freedom}.
When $A$ is trivial, the code is a conventional subspace code.

\emph{Stabilizer codes}~\cite{gottesman:1996:stabilizer} are defined in systems composed of a number of physical qubits. 
The subspace $C$ is defined in terms of a \emph{stabilizer $\stab$}, a subgroup of the Pauli group $\pauli$ with $-\id \not 
\in\stab$.
In particular $C$ is the subspace with projector  
\begin{equation}\label{eq:P}
P=\prod_{s\in \stab}\frac{1+s}2.
\end{equation}
A \emph{gauge group} $\gauge\subseteq\pauli$ with $\stab$ as its center (up to phases) fixes the decomposition $C=A\otimes B$~\cite{poulin:2005:stabilizer}.
The elements of $\gauge$ generate the full algebra of operators on the gauge subsystem $A$, and act trivially on the logical subsystem $B$.
The elements of the group $\cent\gauge$, the centralizer of $\gauge$ in $\pauli$, generate the full algebra on $B$.
They are called \emph{bare logical operators}, as opposed to the \emph{dressed logical operators} in $\cent\stab$, which might involve a nontrivial action in $A$.
Bare logical operators that are equivalent up to stabilizers have the same action: $\logical\subseteq\cent\gauge$ will denote a set of representatives (up to phases) of the quotient $\cent\gauge/\stab$.
Nontrivial dressed logical operators, \emph{i.e.} the elements of $\cent\stab-\gauge$, produce a logical error without leaving any trace in the syndrome.
The \emph{code distance} is the integer
\begin{equation}\label{eq:distance}
d = \min_{E\in\cent\stab-\gauge} |\supp E|
\end{equation}
with $\supp E$ the set of qubits where $E$ acts nontrivially.

\subsubsection*{Error syndrome}

Error detection amounts to measure a generating set  $\stab_0$ of the stabilizer group $\stab$.
The \emph{error syndrome} is the set $\sigma\subseteq \stab_0$ of stabilizer generators that yield a negative eigenvalue.
Thus, if the syndrome $\sigma$ is recovered the system is projected via
\begin{equation}
P_\sigma:=\prod_{s\in\stab_0-\sigma} \frac {1+s} 2 \prod_{s\in \sigma} \frac {1-s} 2.
\end{equation}
For encoded states the syndrome is trivial since $P=P_\emptyset$.
If the error $E\in\pauli$ affects an encoded state, the resulting syndrome is $\synd E$, defined by 
\begin{equation}\label{eq:synd}
s\in \synd E \iff sE=-Es.
\end{equation}
The syndrome is a group homomorphism, \emph{i.e.}
\begin{equation}\label{eq:synd_prop}
\synd (E_1E_2) =\synd E_1+\synd E_2.
\end{equation}
For every syndrome $\sigma$ an operator $\corr \sigma\in\pauli$ must be chosen that will take the system back to the code subspace, \emph{i.e.} it must satisfy
\begin{equation}\label{eq:corr}
\synd {\corr \sigma}=\sigma.
\end{equation}

Not every subset $\sigma\subseteq \stab_0$ corresponds to an error syndrome.
In particular, $\sigma$ is a syndrome if an only if $P_\sigma\neq 0$.
When this is satisfied $\sigma$ is said to be \emph{valid} and $\corr \sigma$ is defined.
If $\sigma$ and $\sigma'$ are valid, then so is $\sigma+\sigma'$ due to \eqref{eq:synd_prop}. 
This will be relevant below for the following reason.
In presence of noise, instead of the correct syndrome $\sigma$ measuring the stabilizer generators will yield in general a non-valid set $\sigma+\omega$. 
Given $\sigma+\omega$ there are many $\omega_0\subseteq \stab_0$ such that
\begin{equation}\label{eq:close_S} 
\sigma + \omega + \omega_0
\end{equation}
is valid.
Importantly, it is possible to choose $\omega_0$ in such a way that it only depends on $\omega$, not $\sigma$.
The reason is that $\omega+\omega_0$ is valid if and only if \eqref{eq:close_S} is.
The syndrome \eqref{eq:close_S} will be called the effective syndrome, and $\omega+\omega_0$ the effective wrong syndrome.

\subsubsection*{Gauge syndrome}

Alternatively, one might define an error syndrome $\sigma$ as a group morphism
\begin{equation}
\rho_\sigma:\stab \longrightarrow \mathbf Z_2.
\end{equation}
The two definitions are exchangeable: 
the morphism maps the elements of $\sigma$ to 1 and the rest of elements of $\stab_0$ to 0.
A \emph {gauge syndrome} can be similarly defined~\cite{bombin:2013:structure} as a subset $\gamma$ of a set $\gauge_0$ of generators of the gauge group, or as a group morphism
\begin{equation}
\rho_\gamma:\gauge \longrightarrow \mathbf Z_2
\end{equation}
such that $\rho_\gamma(\alpha\id)=0$ for any complex number $\alpha$ but otherwise analogous to the stabilizer case: the morphism maps the elements of $\gamma$ to 1 and the rest of elements of $\gauge_0$ to 0.
The definitions in terms of morphisms are convenient because they make it clear that every gauge syndrome is also a stabilizer syndrome, obtained by restricting the domain of $\rho_\gamma$ to $\stab$.
That $\sigma$ is the error syndrome of the gauge syndrome $\gamma$ is denoted
\begin{equation}\label{eq:err}
\sigma= \err \gamma.
\end{equation}
For some codes it is possible to perform a series of measurements of the gauge generators that provide a useful gauge syndrome, \emph{i.e.} one that yields the right error syndrome.
\emph{E.g.} CSS codes have $X$-type and $Z$-type gauge generators that can be measured in two separate rounds to provide a useful gauge syndrome.

Given any subset of stabilizer generators $\sigma\subseteq\stab_0$, consider the map $\tilde \rho_\sigma$ from $\stab_0$ to $\mathbf Z_2$ mapping the elements of $\sigma$ to $1$ and the rest to $0$. 
The set $\sigma$ is an error syndrome precisely when there exists a morphism $\rho_\sigma$ as defined above that extends the map $\tilde \rho_\sigma$. 
This approach to error syndromes extends immediately to gauge syndromes, so that one can distinguish between valid and non-valid subsets $\gamma\subseteq\gauge_0$.
The above discussion on wrong syndrome measurements also extends to the gauge case.
Namely, since $\rho_{\gamma+\gamma'}=\rho_\gamma+\rho_{\gamma'}$ (point by point addition), the set $\gamma+\gamma'$ is valid if $\gamma$ and $\gamma'$ are.
In presence of noise, instead of the correct gauge syndrome $\gamma$ measuring the gauge generators will yield in general a non-valid set $\gamma+\delta$.
As in the case of the error syndrome case, it is possible to choose some $\delta_0\subseteq \gauge_0$ that only depends on $\delta$ and such that
\begin{equation}\label{eq:close_G} 
\gamma + \delta + \delta_0
\end{equation}
is valid.
The effective error syndrome is
\begin{equation}\label{eq:eff_syndrome} 
\err \gamma + \omega,\qquad \omega:=\err (\delta+\delta_0).
\end{equation}
Here $\omega$ is, unlike in~\eqref{eq:close_S}, the effective wrong syndrome, \emph{i.e.} no $\omega_0$ is needed.
For an alternative perspective on gauge syndromes as Pauli operators, see~\cite{bombin:2013:structure}.

\subsubsection*{Channels}

In what follows all the operations will take the form
\begin{equation}\label{eq:channels}
\rho\to\sum_i K_i \,\chan Q_i(\rho)\, K_i^\dagger,\qquad K_i=p_i^{1/2} E_iP_i,
\end{equation}
where $\chan Q_i$ is a CPTP map with Kraus operators in the gauge algebra, $p_i\geq 0$, $E_i\in\pauli$ and each $P_i$ is the sum of any number of different projectors $P_\sigma$.
This is a quantum channel if
\begin{equation}
\sum_i p_i P_i=\id.
\end{equation}
Introducing the channels $\chan Q_i$ makes it possible to avoid a detailed description of operations affecting only the gauge degrees of freedom, \emph{e.g.} the extraction of the gauge syndrome.
In fact, operations of the form \eqref{eq:channels} will be denoted $\sset{K_i}$, thus ignoring the action on gauge qubits entirely.
This is consistent with composition, in the sense that
\begin{equation}
\sset{K_i}\circ\sset {L_j}=\sset{K_iL_j}_{i,j}.
\end{equation} 
The choice \eqref{eq:channels} is rather artificial, but it is rich enough to represent error correction and non-trivial forms of noise.
In particular, the code subspace $C$ can be characterized by the projection
\begin{equation}\label{eq:C}
\chan C :=\sset P,
\end{equation}
and ideal error correction takes the form
\begin{equation}\label{eq:recovery}
\chan R_0=\sset{\corr \sigma P_\sigma}.
\end{equation}

\subsection{Noise characterization}

The purpose of this section is to begin the characterization of the noisy channel classes $\cchan N_{\tau,\epsilon}$ in terms of properties, a task that will be completed in the next section. 
Rather than general channels, any given class will consist of Pauli channels, \emph{i.e.} channels of the form
\begin{equation}\label{eq:noise}
\chan E=:\pchan {p_{\chan E}(E)}E_{E\in\pauli}.
\end{equation}
Pauli channels are characterized by a distribution $p_{\chan E}$ over Pauli errors and include \emph{e.g.} depolarizing noise on individual qubits.

Since the parameter $\epsilon$ intends to characterize logical noise, it is natural to relate it to the failure probability of ideal error correction on a given channel.
The net effect of applying ideal error correction after the noisy channel $\chan E$ affects encoded states is
\begin{equation}\label{eq:net}
\chan R_0\circ \chan E\circ \chan C=\pchan {p_{\chan E}(E)}{\corr{\synd E}EP}_{E\in\pauli}.
\end{equation}
Every $E\in\pauli$ decomposes uniquely as 
\begin{equation}\label{eq:decompose}
E=\corr {\synd E} G L, \qquad G\in\gauge, \qquad L\in\logical.
\end{equation}
\emph{I.e.} each $E$ is uniquely identified by the correction operator for its error syndrome, a logical representative and a gauge operator.
Applying this decomposition to \eqref{eq:net} yields
\begin{equation}
\chan R_0\circ \chan E\circ \chan C=\pchan {p_{\chan E}(\corr \sigma G L)}{LP}_{\sigma,G,L},
\end{equation}
Thus ideal error correction fails with probability
\begin{equation}
\fail {\chan E}:= \sum_{L\neq \id}\sum_{\sigma,G} p_{\chan E}(\corr \sigma G L).
\end{equation}
This failure probability quantifies the logical noise and thus will be a central figure of merit.

The parameter $\tau$ intends to characterize the syndrome distribution when a channel in the class $\cchan N_{\tau,\epsilon}$ is applied to an encoded state.
For a Pauli channel $\chan E$ the distribution is
\begin{equation}\label{eq:synd_distr}
q_{\chan E}(\sigma):=\sum_{G,L} p_{\chan E}(\corr \sigma GL).
\end{equation}
It will be convenient to encapsulate it as the channel
\begin{equation}
\rchan E:=\pchan {q_{\chan E}(\sigma)}{\corr \sigma}_\sigma= \pchan {p_{\chan E}(E)}{\corr {\synd E}}_E,
\end{equation}
which has the same syndrome distribution as $\chan E$ and satisfies $\fail{\rchan E}=0$.
In other words, the channel $\rchan E$ differs from the channel $\chan E$ in that $\rchan E$ applies correctable errors $E'$ instead of the general errors $E$ applied by $\chan E$, in particular via the correspondence $E'=\corr{\synd E}$.
Thus $\chan E$ and $\rchan E$ only differ by their action on encoded qubits, which is trivial in the case of $\rchan E$.

Together, $\fail {\chan E}$ and $\rchan E$ contain a good deal of information about the noisy channel $\chan E$.
This is exemplified by the following result.
\begin{lem}\label{lem:comp}
For any Pauli channels $\chan E$, $\chan D$ channel composition satisfies 
\begin{align}
&\qquad\reduce {\chan E \circ \chan D}=\reduce {\rchan E\circ \rchan D}, \label{eq:comp1}\\
\fail {\chan E\circ \chan D}&\leq \fail {\chan E}+\fail {\chan D}+\fail {\rchan E\circ \rchan D}, \label{eq:comp2}\\
\fail {\rchan E\circ \rchan D}&\leq \fail {\chan E}+\fail {\chan D}+\fail {\chan E\circ \chan D}.\label{eq:comp2bis}
\end{align}
\end{lem}

\begin{sketchproof}
The equality \eqref{eq:comp1} can be readily checked. 
As for the inequality \eqref{eq:comp2}, consider the composition of two Pauli errors
\begin{equation}\label{eq:twoerrors}
E_1E_2=\corr {\sigma_1+\sigma_2}GL, \qquad E_i=\corr {\sigma_i} G_i L_i.
\end{equation}
If $L\neq 1$ at least one the next must hold: $L_1\neq 1$ or $L_2\neq 1$ or 
\begin{equation}\label{eq:bad}
\corr {\sigma_1+\sigma_2}\corr {\sigma_1}\corr {\sigma_2}\not\in\gauge.
\end{equation}
Each of these cases corresponds to one of the terms in \eqref{eq:comp2}.
A similar reasoning gives~\eqref{eq:comp2bis}.
Namely, \eqref{eq:bad} holds for $E_i$ as in \eqref{eq:twoerrors} only if either $L_1\neq 1$ or $L_2\neq 1$ or $L\neq 1$, and each case gives a term in \eqref{eq:comp2bis}.
\end{sketchproof}

It is now possible to write down a first property that any parametrized collection of classes $\cchan N_{\tau,\epsilon}$ should satisfy:
\begin{description}[style=sameline]
\item [1.]\it
A channel $\chan E$ belongs to the class $\cchan N_{\tau,\epsilon}$ if and only if
\begin{equation}\label{eq:axiom1}
\fail {\chan E}\leq \epsilon \qquad\text{and} \qquad\rchan E\in \cchan N_{\tau,0}.
\end{equation}
\end{description}
In particular, the collection of classes $\cchan N_{\tau,\epsilon}$ is fixed by the classes with no logical noise $\cchan N_{\tau,0}$.
A second property is the following.
\begin{description}[style=sameline]
\item [2.]\it
For any $\tau$ and $\tau'$, there exists $\delta$ such that
\begin{equation}\label{eq:axiom2}
\cchan N_{\tau,0}\circ \cchan N_{\tau',0}\subseteq \cchan N_{\tau+\tau',\delta}.
\end{equation}
\end{description}
It is implicitly assumed that a given parametrized collection of classes $\cchan N_{\tau,\epsilon}$ actually has a third parameter, the system size: it needs not be explicit in equations because it is the same for all the objects involved.
Therefore $\delta$ not only depends on $\tau$ and $\tau'$, but also on the system size.
The interesting scenario is one in which errors do not pile up catastrophically and logical errors can be made arbitrarily small.
This means that there should be a threshold value for $\tau+\tau'$ below which $\delta$ can be made as small as desired solely by increasing the system size.
Moreover, this should still be true when we compose channels with non-zero logical noise.
This is the content of the next corollary, a direct consequence of the results (\ref{eq:comp1}, \ref{eq:comp2}) in lemma~\ref{lem:comp}.
\begin{cor}
For a collection of classes $\cchan N_{\tau,\epsilon}$ satisfying the above two properties and $\delta$ as in \eqref{eq:axiom2}, 
\begin{equation}\label{eq:compN}
\cchan N_{\tau,\epsilon}\circ \cchan N_{\tau',\epsilon'}\subseteq \cchan N_{\tau+\tau',\epsilon+\epsilon'+\delta}.
\end{equation}
\end{cor}

\subsection{Local noise}\label{sec:local_noise}

The errors caused by interactions with the environment or within the implementation of local circuits are expected to have a local nature.
To model locality it will be convenient to make use of the following concept.
\begin{defn}\label{def:bounded}
Let $p(A)$ be a probability distribution over subsets $A\subseteq B$ of some set $B$.
Given some $\alpha>0$ the distribution $p$ is $\alpha$-bounded if for every $A\subseteq B$
\begin{equation}\label{eq:tilde}
\tilde p(A):=\sum_{A'\supseteq A} p(A') \leq \alpha^{|A|}.
\end{equation}
\end{defn}
To characterize the locality of a Pauli channel $\chan E$, rather than considering the probability distribution $p_{\chan E}(E)$ over Pauli errors $E$ it is enough to have the corresponding probability distribution for their support
\begin{equation}\label{eq:local} 
p_{\chan E}'(R):=\sum_{E\suchthat \supp E= R} p_{\chan E}(E),
\end{equation}
where $R$ is any set of qubits.
\begin{defn}
Given a qubit system, $\cchan L_{\lambda}$ is the set of Pauli channels $\chan E$ for which the distribution $p'_{\chan E}$ of \eqref{eq:local} is $\lambda$-bounded.
\end{defn}
The classes of channels $\cchan L_{\lambda}$ model local noise: errors that affect a specific large set of qubits are highly unlikely.

The aim of the third and last property for parametrized collections of classes $\cchan N_{\tau,\epsilon}$ is to guarantee their compatibility with local noise:\begin{description}[style=sameline]
\item [3.]\it
For every $\tau>0$ there exist $\lambda>0$ and $\epsilon$ such that
\begin{equation}\label{eq:axiom3}
\cchan L_\lambda\subseteq \cchan N_{\tau,\epsilon}.
\end{equation}
\end{description}
In other words, the `temperature' $\tau$ caused by local noise can be made arbitrarily small by reducing the intensity of the local noise.
The logical noise parameter can be taken to be $\epsilon=f(\lambda)$ with the function $f$ satisfying
\begin{equation}\label{eq:f}
\forall \chan E\in\cchan L_\lambda,\qquad \fail {\chan E} \leq f(\lambda).
\end{equation}
As usual $f$ is implicitly dependent on the system size.
In the interesting scenarios $f$ decreases rapidly as the system size increases for $\lambda$ below a threshold value.

Later several different collections of classes $\cchan N_{\tau,\epsilon}$ of noisy channels will be needed.
The most straightforward collection will play a role when discussing 3D gauge color codes.
It will be denoted $\cchan N_{\tau,\epsilon}^\mathrm{loc}$ and it just models standard local noise in the following sense.
\begin{defn}
Given a qubit system and a stabilizer code,  $\cchan N_{\tau,\epsilon}^\mathrm{loc}$ is the set of Pauli channels $\chan E$ such that  $\fail {\chan E}\leq \epsilon$ and  there exists $\chan D\in  \cchan L_{\tau}$ with
\begin{equation}\label{eq:def_Nlocal}
\rchan E = \rchan D.
\end{equation}
\end{defn}

\begin{prop}
If \eqref{eq:f} holds, the collection of classes $\cchan N_{\tau,\epsilon}^\mathrm{loc}$ satisfies the above properties 1-3 with 
\begin{align}\label{eq:delta_taus}
\delta &= f(\tau)+f(\tau')+f(\tau+\tau'),\\
\lambda&=\tau, \qquad \epsilon = f(\tau).
\end{align}
\end{prop}

\begin{proof}
First notice that
\begin{equation}
\cchan L_\tau\circ\cchan L_{\tau'}\subseteq \cchan L_{\tau+\tau'},
\end{equation}
because for $\chan D\in \cchan L_{\tau}, \chan D'\in \cchan L_{\tau'}$,
using the tilde notation of~\eqref{eq:tilde} on the probability distributions of error supports,
\begin{multline}\label{eq:sum_cases}
\tilde p'_{\chan D\circ \chan D'}(R)\leq \sum_{S\sub R}\tilde p_{\chan D}(S)\tilde p_{\chan D'}(R-S)
\leq\\
\leq \sum_{n=0}^{|R|} {|R| \choose n}\tau^n(\tau')^{|R|-n}= (\tau+\tau')^{|R|}.
\end{multline}

Only the property 2 is non-trivial, and it amounts to
\begin{equation}\label{prop2local}
\chan E\in \cchan N_{\tau,0}^\mathrm{loc}, \chan E'\in \cchan N_{\tau',0}^\mathrm{loc}\quad\Longrightarrow\quad
\chan E\circ\chan E'\in \cchan N_{\tau+\tau',\delta}^\mathrm{loc}
\end{equation}
For such $\chan E$ there exists $\chan D\in  \cchan L_{\tau}$ as in \eqref{eq:def_Nlocal}, and analogously for $\chan E'$ and some $\chan D'\in  \cchan L_{\tau'}$.
The inequality \eqref{eq:comp2bis} takes here the form
\begin{multline}\label{eq:comp2bis_app}
\fail {\chan E\circ \chan {E'}} = \fail {\rchan E\circ \rchan E'}= \fail {\rchan D\circ \rchan {D'}} \leq 
\\ \leq
\fail {\chan D}+\fail {\chan D'}+\fail {\chan D\circ \chan D'}\leq \delta.
\end{multline}
To prove \eqref{prop2local}, it only rests to observe that
\begin{equation}
\reduce {\chan E \circ \chan E'}=\reduce{\reduce {\chan E} \circ \reduce{\chan E'}}=\reduce{\reduce {\chan D} \circ \reduce{\chan D'}}=\reduce {\chan D \circ \chan D'}.
\end{equation}
\end{proof}

It might not be obvious why there is any need to consider classes of noise other than $\cchan N_{\tau,\epsilon}^\mathrm{loc}$. Actually in some cases the physics of the problem (as exemplified by the Ising model in section \ref{sec:confinement}) will naturally give rise to some class of noise for which single-shot error correction is possible.

\subsection{Local codes}\label{sec:local_codes}

A stabilizer code with a fixed set of stabilizer generators induces a notion of connectedness:
qubits are nodes of a graph, with any two linked if they are both in the support of some generator.
A family of stabilizer codes is local if the degree of the nodes of all such graphs is bounded~\footnote{Codes with this property are called LDPC: low-density parity check code.} and the family includes codes with arbitrarily large distance $d$.
Notices that for such a family the number of stabilizer generators is bounded, namely
\begin{equation}\label{eq:bound_gen}
|\stab_0|=O(n),
\end{equation}
where $n$ is the number of qubits of a given code in the family.

For local noise it is natural, but not optimal, to choose $\corr \sigma$ so that $|\supp {\corr \sigma}|$ is minimal.
With this choice any family of local codes displays a threshold.
In particular, as shown in \cite{gottesman:2013:overhead}, for a given family there exists some $\lambda_0>0$ such that if $\lambda<\lambda_0$ then the bound \eqref{eq:f} on the logical error probability for channels in $\cchan L_\lambda$ holds with
\begin{equation}\label{eq:local_codes}
f(\lambda) \propto n \left (\frac \lambda{\lambda_0}\right)^{d/2},
\end{equation}
where $n$ is the number of qubits and the proportionality constant is system size independent.

Notable among local stabilizer codes are topological stabilizer codes, see \emph{e.g.}~\cite{bombin:2013:topological}.
In topological codes locality is geometric, in the sense that qubits are placed on a lattice and stabilizer generators only involve qubits on a given ball of bounded radius.
Typically the distance $d$ scales as a polynomial on the linear size of the lattice, and so does the number of qubits $n$, so that \eqref{eq:local_codes} yields an exponential suppression of logical errors with the system size.
It is maybe also due to their strong physical flavor that topological stabilizer codes display other very interesting features, such as the fact that classical processing for error correction is often efficient~\cite{dennis:2002:tqm,duclos:2010:fast},
that error thresholds are high~\cite{dennis:2002:tqm,bombin:2012:strong},
and that computations can be performed in many different ways, some conventional and some of a topological nature~\cite{dennis:2002:tqm,bombin:2006:2DCC,bombin:2007:3DCC,raussendorf:2007:deformation,bombin:2009:deformation,bombin:2010:twist,horsman:2012:surgery,landahl:2014:surgery}.

\subsection{Fault-tolerant error correction}

In modelling noisy recovery, a phenomenological approach that captures the essential difficulties involved in fault-tolerant error correction will be enough.
For local codes, the only ones considered in this work, a minimal model could include faulty measurement outcomes with a probability independent of the outcome, and local noise \eqref{eq:local} following the ideal recovery operation~\cite{dennis:2002:tqm}.
The latter form of noise can be disregarded since it is common to any noisy operation and, as long as \eqref{eq:axiom3} and \eqref{eq:compN} hold, does not give rise to any difficulty.
An implicit assumption here is that the form of the original noise $\cchan N_{\tau,\epsilon}$ to be corrected is preserved under local operations up to a change in the parameters $\tau,\epsilon$ (such that $\epsilon$ remains the same in the large system limit and $\tau$ remains close to zero if it was so). 
This is indeed the case for the classes of channels considered here.

These considerations lead to model noisy recovery as
\begin{equation}\label{eq:noisy_R}
\chan R=\pchan{q_{\chan R}(\omega)}{\corr {\sigma+\omega} P_\sigma}_{\sigma,\omega},
\end{equation}
with $\sigma,\omega$ syndromes: the outcome $\sigma'=\sigma+\omega$ differs at the generators in $\omega$ from the actual syndrome $\sigma$ with probability $q_{\chan R}(\omega)$.
The syndrome $\sigma+\omega$ is the effective error syndrome, and $\omega$ is the effective wrong syndrome.
A class  of noisy recovery channels $\cchan R_{\eta}$ will be composed of channels of the form \eqref{eq:noisy_R} satisfying some condition with parameter $\eta$.

Eventually the goal is to understand the net effect of errors and noisy recovery, as in \eqref{eq:goal}.
The following result will be central to this endeavor.
\begin{lem}\label{lem:central}
Suppose that for some $\eta$ and $\tau'$
\begin{equation}\label{eq:goalR}
\chan R\in \cchan R_{\eta}\quad\Longrightarrow\quad\echan R\in \cchan N_{\tau',0},
\end{equation}
where $\ceff {\chan R}$ is the `effective' channel
\begin{equation}\label{eq:effective}
\echan R:=\pchan {q_{\chan R}(\omega)}{\corr \omega}_\omega.
\end{equation}
Then \eqref{eq:goal} holds for any $\epsilon$ and $\tau$ with $\delta$ as in \eqref{eq:axiom2}.
\end{lem}

\begin{proof}
For any Pauli channel $\chan E$ there exists a Pauli channel $\chan F$ such that
\begin{equation}
\chan R\circ\chan E\circ\chan C=\chan F\circ \chan C
\end{equation}
and with the following properties:
\begin{equation}\label{eq:eff}
\fail {\chan F}=\fail {\echan R\circ \chan E},\qquad \rchan F=\ceff{\chan R}.
\end{equation}
Indeed, one can take 
\begin{equation}\label{eq:F1}
\chan F= \pchan{p(\omega,E)}{\corr {\omega+\synd E}E}_{\omega,E},
\end{equation}
where $p(\omega,E):=q_{\chan R}(\omega)p_{\chan E}(E)$.
To recover the second equality in \eqref{eq:eff} rewrite this as
\begin{equation}\label{eq:F2}
\chan F= \pchan{p(\omega,E)}{\corr \omega(\corr \omega\corr {\omega+\synd E}E)}_{\omega,E}
\end{equation}
where the operator in parenthesis is logical.
As for the first equality, just compare \eqref{eq:F2} with
\begin{equation}\label{eq:F3}
\echan R\circ\chan E =\pchan {p(\omega,E)}{\corr {\omega+\synd E} ( \corr {\omega+\synd E}\corr \omega E)}_{\omega,E}.
\end{equation}
Putting together \eqref{eq:goalR}, \eqref{eq:eff}, and \eqref{eq:compN} yields as desired:
\begin{equation}
\chan R\in \cchan R_{\eta}, \chan E\in \cchan N_{\tau,\epsilon}\quad\Longrightarrow \quad\chan F\in \cchan N_{\tau',\epsilon+\delta}.
\end{equation}
\end{proof}

Notice that $\echan R$ is characterized by a syndrome distribution, which in turn can be assigned a temperature $\tau'$ that describes the confinement of the syndromes.
By comparing \eqref{eq:effective} with \eqref{eq:noisy_R} it is apparent that $\tau'$ is the temperature of the wrong syndrome outcomes.
Thus, as stated before, the parameter $\eta$ controls the level of confinement of wrong syndrome outcomes.
Remarkably, the success of single-shot error correction boils down to the existence of such a confinement.
As expected from the examples in section~\ref{sec:noisy}, the residual temperature $\tau'$ coincides with the temperature describing the wrong syndrome outcomes.
In the ideal scenario, for every $\tau'>0$ there will be some $\eta>0$ such that \eqref{eq:goalR} holds.
In that case the residual temperature $\tau'$ can be made arbitrarily small by improving the precision of the syndrome extraction operation.
This will be the case in all the examples examined here.

\subsection{Connectivity combinatorics}

This section provides some notation and simple technical lemmas that will be used repeatedly, but only in proofs.
The main tools will be graphs of bounded maximum degree and connected clusters (sets of nodes).

\begin{defn}
Consider a graph with node set $\Gamma$, a cluster $V\subseteq \Gamma$ and some $k\geq 0$. $C_{s,k}(V)$ is a subset of the powerset of $\Gamma$ with the following property: a cluster $W$  is an element of $C_{s,k}(V)$ if there exists a cluster $W_0$ 
with
\begin{equation}
|W_0|=s, \qquad V\cup W\subseteq W_0,
\end{equation}
and such that each connected component $W_\mathrm c$ of $W_0$ satisfies
\begin{equation}
V\cap W_\mathrm c\neq \emptyset, \qquad |W\cap W_\mathrm c|\geq k|W_\mathrm c|.
\end{equation}
\end{defn}
 
It is convenient to define also
\begin{equation}
C_s(V):=C_{s,1}(V), \qquad C_s:=\bigcup_{v\in \Gamma} C_s(\sset v).
\end{equation}
\emph{I.e.} $C_s(V)$ is the collection of clusters $V_0\supseteq V$ of size $s$ such that each of their connected components contains an element of $V$, and $C_s$ is the collection of connected clusters of size $s$.
For every $W\in C_{s,k}(V)$ there exists some $W_0\in C_{s}(V)$ with $W\subseteq W_0$. 
Thus $|C_{s,k}(V)|\leq 2^s|C_s(V)|$.
Moreover (see e.g. \cite{gottesman:2013:overhead, bombin:2013:self}, in particular lemma 2 in~\cite{gottesman:2013:overhead}), for graphs of bounded maximum degree there exists some constant $\gamma$ (depending only on the bound) such that $|C_s(V)|\leq \gamma^s$. 
These observations provide the following lemma.

\begin{lem}\label{lem:C}
Consider a family of graphs with bounded maximum degree and some $k\geq 0$.
There exist some $\gamma> 0$ such that for every integer $s$, every graph in the family and every cluster $V$ in the graph
\begin{equation}\label{eq:bound_C}
C_{s,k}(V)\leq \gamma^{s}.
\end{equation}
\end{lem}

Notice that $\gamma$ depends on $k$.
For $k=1$,
\begin{equation}\label{eq:bound_Cs}
|C_s|\leq \sum_{v\in\Gamma}|C_s(\sset v)| \leq |\Gamma|\gamma^{s}. 
\end{equation}

\begin{lem}\label{lem:bounded}
Consider a family of graphs with bounded maximum degree and some $k>0$.
Each graph has a set of nodes $\Gamma$ and comes equipped with (i) two subsets of nodes $\Gamma_i\subseteq \Gamma$, $i=1,2$, and (ii) a function $f:\Gamma_1\longrightarrow\Gamma$ such that for every cluster $V\subseteq \Gamma_1$ 
and for every connected component $V_{\text c}$ of the cluster $f(V)$
\begin{equation}\label{eq:connectivity}
|V\cap V_{\text c}|\geq k |V_{\text c}|.
\end{equation}

There exists some $\alpha_0>0$ such that for every $\alpha$ with $0<\alpha<\alpha_0$ the following holds.
If a probability distribution $p_1(V)$ over clusters $V\subseteq\Gamma_1$ is $\alpha$-bounded, then the probability distribution $p_2(V')$ over clusters $V'\subseteq\Gamma_2$ defined by
\begin{equation}\label{eq:p2}
p_2(V'):=\sum_{V\suchthat f(V)\cap\Gamma_2= V'} p_1(V).
\end{equation}
is $\beta$-bounded with
\begin{equation}
\beta :=  \frac {(\alpha/\alpha_0)^{k}}{1-(\alpha/\alpha_0)^{k}},
\end{equation}
\end{lem}

\begin{proof}
The goal is to show that if $p_1$ is $\alpha$-bounded then
\begin{equation}\label{eq:goal_p2}
\tilde p_2(V') \leq \beta^{|V'|},
\end{equation}
where $V'\subseteq \Gamma_2$ and, according to definitions~\eqref{eq:tilde} and~\eqref{eq:p2}, 
\begin{multline}\label{eq:tildep2}
\tilde p_2(V') 
 =\sum_{V''\suchthat V'\subseteq V''\subseteq \Gamma_2}\,\,\,  \sum_{V\suchthat V''= f(V)\cap \Gamma_2} p_1(V) =\\
 = \sum_{V\suchthat V'\subseteq f(V)\cap \Gamma_2\subseteq\Gamma_2} p_1(V) =
 \sum_{V\suchthat V'\subseteq f(V)} p_1(V)
\end{multline}

Take $\alpha_0=\gamma^{-1/k}$ with $\gamma$ as in lemma \ref{lem:C}, applied to the family of graphs $\Gamma$ and the given value of $k$.
Denote by $C_{s,k}^1(V')$ the intersection of $C_{s,k}(V')$ and the powerset of $\Gamma_1$.
It satisfies
\begin{equation}
|C_{s,k}^1(V')|\leq \alpha_0^{-ks}.
\end{equation}

Consider $V\subseteq\Gamma_1$ and $V'\subseteq\Gamma_2$ such that $V'\subseteq f(V)$, as in the terms of~\eqref{eq:tildep2}.
Let $W_0$ be the union of the connected components of $f(V)$ that contain some element of $V'$, and define
\begin{equation}
W:=W_0\cap V , \qquad s:=|W_0|.
\end{equation}
A connected component $W_\mathrm c$ of $W_0$ is also a connected component of $f(V)$ and thus
\begin{equation}
|W\cap W_\mathrm c|=|W_0\cap V \cap W_\mathrm c|=|V \cap W_\mathrm c|\geq k|W_\mathrm c|,
\end{equation}
where the inequality is the property~\eqref{eq:connectivity}.
Moreover, $V'\cup W \subseteq W_0$ and thus $W\in C_{s,k}^1(V')$.
Noting that $|W|\geq k|W_0|=ks$, the probability \eqref{eq:tildep2} can be bounded as follows for $|V'|>0$:
\begin{multline}
\tilde p_2(V') \leq \sum_{s\geq |V'|}\sum_{W \in C_{s,k}^1(V')} \tilde p_1(W)\leq 
\sum_{s\geq |V' |}\left(\frac \alpha{\alpha_0}\right)^{ks} =
\\ =
\frac {(\alpha/\alpha_0)^{k|V'|}}{1-(\alpha/\alpha_0)^{k}} \leq 
\frac {(\alpha/\alpha_0)^{k|V'|}}{(1-(\alpha/\alpha_0)^{k})^{|V'|}}.
\end{multline}
This gives \eqref{eq:goal_p2} because $\tilde p_2(\emptyset)\leq 1$.
\end{proof}

\section{Self-correction}\label{sec:self-correction}

This section discusses the connection between self-correction and single-shot error correction in conventional topological stabilizer codes.

\subsection{Characterization}\label{sec:characterization}

For a given topological stabilizer code with stabilizer generators $\stab_0$, there is a local quantum Hamiltonian model 
\begin{equation}\label{eq:Hamiltonian}
H=-J\sum_{s\in \stab_0} s,\qquad J>0.
\end{equation}
The ground state is the code subspace, and excitation configurations can be identified with error syndromes, \emph{i.e.} valid subsets $\sigma\sub \stab_0$.

For some code families the system \eqref{eq:Hamiltonian} is self-correcting: 
in the thermodynamic limit the memory remains indefinitely protected at finite temperatures, without any active error correction.
Known examples~\cite{bombin:2013:self} rely on the connectivity of excitations and satisfy several properties that will allow single-shot error correction.
The purpose of this section is to enumerate these properties, which are not true for general topological stabilizer codes but hold for known self-correcting codes.
The following sections will show how single-shot error correction follows by assuming these properties.

Individual excitations (elements of $\stab_0$) can be regarded as the nodes of some given graph that is consistent with the geometry of the lattice, \emph{i.e.} linked nodes have a bounded distance.
This guarantees that the graph has bounded maximum degree, \emph{i.e.} each excitation is connected only to a bounded number of other excitations.
A first crucial property is that the graph can be defined~\cite{bombin:2013:self} so as to satisfy the following:
\begin{description}[style=sameline]
\item [I]
Each connected component of a syndrome is itself a syndrome.
\end{description}
This is, in fact, a slight simplification, see section~\ref{sec:global_constraints}.

The graph of excitations induces another one for qubits.
Namely, two qubits $q_1$, $q_2$ are linked when there exist two linked excitations with supports $s_1$, $s_2$ such that $q_i\in s_1\cup s_2$.
The resulting graph is also consistent with the geometry of the lattice and in particular has bounded maximum degree.
This connectivity of qubits, which is different from the one defined in section~\ref{sec:local_codes}, is designed to satisfy the following property:
\begin{description}[style=sameline]
\item [II]
Pauli errors with mutually disconnected support have mutually disconnected syndromes.
\end{description}
Two clusters (sets of nodes) $A,B$ are mutually disconnected if no element of $A$ is linked to an element of $B$ or belongs to $B$.

According to property I it is possible to choose the correcting operators such that for any mutually disconnected $\sigma, \sigma'$, 
\begin{equation}\label{eq:separate_corr}
\corr {\sigma+\sigma'} = \corr \sigma \corr \sigma'.
\end{equation}
In other words, each connected component of the syndrome can be corrected separately.
According to property II, if~\eqref{eq:separate_corr} holds then a bad error, \emph{i.e.} one with nontrivial $L$ in \eqref{eq:decompose}, must be such that its restriction to one of its connected components is also bad.
Such a choice of correcting operators is applied in Ref.~\cite{bombin:2013:self} to known self-correcting stabilizer systems, in particular choosing the $\corr \sigma$ as local as possible.
The choice is such that there exist some $b,c, i, j> 0$ satisfying the following~\cite{bombin:2013:self}:
\begin{description}[style=sameline]
\item [III]
A Pauli error $E$ is not bad if each of the connected components of its support has at most $bn^i$ qubits.
\item [IV]
If two syndromes $\sigma, \omega$ are such that their union has connected components with at most $cn^{j}$ elements each, then
\begin{equation}\label{eq:cond2}
\corr \sigma\corr \omega\corr {\sigma+\omega}\in\gauge.
\end{equation}
\end{description}
As usual $n$ is the number of qubits in a code of the family.

\subsection{Excitations and noise}\label{sec:excitations}

A tailored class of noisy channels $\cchan N_{\tau,\epsilon}$ is required here.
Recall definition~\ref{def:bounded}.
The probability distribution~\eqref{eq:synd_distr} can be regarded as a distribution $q_{\chan E}(\xi)$ over arbitrary sets of excitations $\xi\subseteq \stab_0$, with $q_{\chan E}(\xi)=0$ for $\xi$ not a syndrome.
\begin{defn}
A Pauli channel $\chan E$ is an element of $\cchan N_{\tau,\epsilon}^\mathrm{exc}$ if $\fail {\chan E}\leq \epsilon$
and 
$q_{\chan E}$ is $\tau$-bounded.
\end{defn}
\begin{prop}
Consider a family of topological stabilizer codes satisfying the properties of section~\ref{sec:characterization}, with $v_1$ the maximum number of qubits in the support of a stabilizer generator and $v_2$ the maximum number of stabilizer generators with support on a qubit.
There exists some system size independent $\tau_0>0$, $\lambda_0>0$ such that the collection of classes $\cchan N_{\tau,\epsilon}^\mathrm{exc}$ satisfies the properties 1-3 of section~\ref{sec:models} with
\begin{align}
\label{eq:delta_sc}
\delta &= |\stab_0| \left(\frac {\tau+\tau'}{\tau_0}\right)^{cn^{j}},\\
\label{eq:taumu_sc}
\tau &= v_1 \lambda^{{1/v_2}},\qquad 
\epsilon = n \left(\frac {\lambda}{\lambda_0}\right)^{bn^i}.
\end{align}
\end{prop}
Before proving these facts, notice that all the required properties are met, \emph{i.e.} (i) $\delta$ and $\epsilon$ can be made arbitrarily small just by increasing the system size as long as $\tau+\tau'<\tau_0$ and $\lambda<\lambda_0$ (recall \eqref{eq:bound_gen}), and (ii) $\tau$ can be made arbitrarily small just by reducing $\lambda$.

\begin{proof}
The first property is obvious.
The second states that:
\begin{equation}
\chan E\in \cchan N_{\tau,0}^\mathrm{exc}, \chan E'\in \cchan N_{\tau',0}^\mathrm{exc}\quad\Longrightarrow\quad
\chan E\circ\chan E'\in \cchan N_{\tau+\tau',\delta}^\mathrm{exc}
\end{equation}
Due to \eqref{eq:synd_prop} and proceeding as in \eqref{eq:sum_cases}
\begin{equation}
\tilde q_{\chan E\circ \chan E'}(\xi)\leq\sum_{\xi' \subseteq \xi} \tilde q_{\chan E}(\xi')\tilde q_{\chan E'}(\xi-\xi')\leq (\tau+\tau')^{|\xi|},
\end{equation}
and thus $\reduce {\chan E\circ \chan E'}$ satisfies the required constraint.
It suffices to bound the probability
\begin{equation}\label{eq:fail_comp}
\fail {\chan E\circ\chan E'}=\sum_{\text {bad }\sigma,\omega} q_{\chan E}(\sigma)q_{\chan E'}(\omega),
\end{equation}
where $\sigma,\omega$ are bad if \eqref{eq:cond2} does not hold.
Take $\tau_0=\gamma^{-1}$ with $\gamma$ as in lemma \ref{lem:C}, applied to the family of graphs of excitations and for $k=1$.
The set of connected syndrome clusters of size $bn^j$ is $B:=C_{bn^j}$.
It satisfies, according to \eqref{eq:bound_Cs},
\begin{equation}
|B|\leq |\stab_0|\,\tau_0^{-cn^j}.
\end{equation}
For $\sigma,\omega$ bad, $\sigma\cup\omega$ must include an element of $B$.
Thus the probability \eqref{eq:fail_comp} can be bounded in agreement with~\eqref{eq:delta_sc} as follows:
\begin{multline}\label{eq:count_prob}
\fail {\chan E\circ\chan E'}\leq \sum_{\xi \in B}\sum_{\xi'\subseteq \xi} \tilde q_{\chan E}(\xi')\tilde q_{\chan E'}(\xi-\xi')\leq 
\\
\leq |B|\, (\tau+\tau')^{cn^j}\leq  |\stab_0| \left(\frac{\tau+\tau'}{\tau_0}\right)^{cn^j}.
\end{multline}

The third property states that $\cchan L_\lambda\subseteq \cchan N_{\tau,\epsilon}^\mathrm{exc}$, \emph{i.e.} that for any $\chan E\in \cchan L_{\lambda}$ and any set of excitations $\sigma$
\begin{equation}\label{eq:re_L}
\tilde q_{\chan E}(\sigma)\leq \tau^{|\sigma|},\qquad \fail {\chan E}\leq \epsilon,
\end{equation}
with $\tau$ and $\epsilon$ as in~\eqref{eq:taumu_sc}.
Given a set of excitations $\sigma$, let $F$ be the collection of functions $f$ that map elements of $\sigma$ to qubits in such a way that for any generator $s\in\sigma$ the qubit $f(s)$ is in the support of $s$.
Clearly
\begin{equation}\label{eq:Ff}
|F|\leq v_1^{|\sigma|},\qquad |f[\sigma]|\geq |\sigma|/v_2,
\end{equation}
where $f[\sigma]$ denotes the image of some $f\in F$.
The first inequality in \eqref{eq:re_L} is obtained using~\eqref{eq:Ff} in
\begin{equation}
\tilde q_{\chan E}(\sigma)\leq \sum_{f\in F} \tilde p_{\chan E}(f[\sigma])\leq \sum_{f\in F} \lambda^{|f[\sigma]|}\leq |F|\lambda^{\min_\sigma{|f[\sigma]|}}.
\end{equation}
Take $\lambda_0=\gamma^{-1}$ with $\gamma$ as in lemma \ref{lem:C}, applied to the family of graphs of qubits and for $k=1$.
The set of connected qubit clusters of size $bn^i$ is $B:=C_{bn^i}$.
It satisfies, according to \eqref{eq:bound_Cs},
\begin{equation}\label{eq:bound_C2}
|B'|\leq n\,\lambda_0^{-bn^i}.
\end{equation}
The support of every bad error $E$ has an element of $B'$ as a subset, and the second inequality in~\eqref{eq:re_L} is obtained using \eqref{eq:bound_C2} in
\begin{equation}
\fail {\chan E}=\sum_{\text {bad }E} p_{\chan E}(E)\leq \sum_{S \in B'} \tilde p_{\chan E}(S)\leq |B'|\lambda^{bn^i}.
\end{equation}
\end{proof}

\subsection {Single-shot error correction}\label{sec:self-ss}

The quantum-local error correction procedure is the following.
First the generators $\stab_0$ are measured, which yields a set of `excitations'.
To model noise in these measurements, assume that instead of the correct set $\sigma$ the measurement yields a set $\sigma+\omega$ with a probability $r(\omega)$ independent of $\sigma$.
As discussed in section~\ref{sec:stabilizer}, some $\omega_0\subseteq\stab_0$ must be chosen so that \eqref{eq:close_S} is valid, and this can be done in such a way that $\omega_0$ only depends on $\omega$.
Here we assume that $\omega_0$ is chosen to have minimal cardinality. 
This choice is not always optimal, but is convenient for deriving results.
Moreover, minimality can be relaxed without affecting the results below, as discussed in section~\ref{sec:complexity}.
This is relevant from a computational perspective, as finding a minimal $\omega_0$ might be too hard.
The error correction procedure finishes with the application of $\corr {\sigma+\omega+\omega_0}$.

A noisy recovery operation $\chan R$  is defined by a distribution $q_{\chan R}$ of wrong effective syndromes as in~\eqref{eq:noisy_R}.
In the procedure just discussed the effective wrong syndrome is $\omega+\omega_0$.
This leads to the following definition for the class of recovery operations, where the superscript $\mathrm{s}$ indicates that the generators $\stab_0$ are measured directly.
\begin{defn}
The class $\cchan R_\eta^\mathrm{s}$ contains those recovery operations $\chan R$ such that
\begin{equation}
q_{\chan R}(\omega')=\sum_{\omega|\omega+\omega_0=\omega'} r(\omega)
\end{equation}
for some probability distribution $r(\omega)$ that is $\eta$-bounded.
\end{defn}
The following result states that single-shot error correction is possible for topological stabilizer codes exhibiting self-correcting features in the sense of section~\ref{sec:characterization}.
\begin{thm}\label{thm:self}
For any family of topological stabilizer codes satisfying the properties of section~\ref{sec:characterization} there exists some $\eta_0>0$ such that 
\begin{equation}\label{eq:goalexc}
\cchan R_{\eta}^\mathrm{s}\circ \cchan N_{\tau,\epsilon}^\mathrm{exc}\circ \chan C\subseteq \cchan N_{\tau',\epsilon+\delta}^\mathrm{exc}\circ \chan C
\end{equation}
holds for $\eta<\eta_0$ and every $\epsilon, \tau$, with $\delta$ as in~\eqref{eq:delta_sc} and
\begin{equation}\label{eq:taup_sc}
\tau' =  \frac {(\eta/\eta_0)^{1/2}}{1-(\eta/\eta_0)^{1/2}}.
\end{equation}
\end{thm}
Notice that the right properties are met, \emph{i.e.} (i) $\tau'$ can be made arbitrarily small just by reducing $\eta$ and  (ii) $\delta$ can be made arbitrarily small just by increasing the system size as long as $\tau+\tau'<\tau_0$.

\begin{proof}
According to lemma~\ref{lem:central} it suffices to show that for every $\chan R\in \cchan R_{\eta}^\mathrm{s}$
\begin{equation}\label{eq:goalRs}
\echan R\in \cchan N_{\tau',0}^\mathrm{exc},
\end{equation}
\emph{i.e.} that the distribution $q_{\chan R}$ is $\tau'$-bounded.
This follows from lemma~\ref{lem:bounded} taking $k=1/2$, $\Gamma=\Gamma_1=\Gamma_2=\stab_0$ and $f(\omega)=\omega+\omega_0$.
In this case \eqref{eq:connectivity} reads
\begin{equation}
2|\omega\cap\omega_\mathrm c|\geq |\omega_\mathrm c|,
\end{equation}
where $\omega_{\text c}$ is any connected component of $\omega+\omega_0$.
This condition is satisfied due to the minimality of $\omega_0$, as follows.
First, notice that
\begin{equation}
|\omega\cap \omega_\mathrm c|=|\omega_\mathrm c|-|\omega_0\cap \omega_\mathrm c|.
\end{equation}
Second, $\omega_{\text c}$ is a syndrome and thus $\omega+\omega_0'$ is a syndrome for $\omega_0':=\omega_0+\omega_{\text c}$, so that by minimality 
\begin{equation}\label{eq:byminimality}
0\leq|\omega_0'|-|\omega_0|=|\omega_\mathrm c|-2|\omega_\mathrm c\cap \omega_0|=2|\omega\cap \omega_\mathrm c|-|\omega_\mathrm c|.
\end{equation}

\end{proof}

\subsection{Global constraints}\label{sec:global_constraints}

In the previous section it has been essential that each connected component of a syndrome is itself a syndrome.
This is, however, not always true.
In the topological stabilizer codes studied in~\cite{bombin:2013:self}, the linear constraints that sets of excitations have to satisfy in order to be syndromes can be arranged in two types: local and global.
Local constraints can be represented via the connectivity of excitations: if a set of excitations satisfies them, so does each connected component of the set.
This is not true for global constraints, which depend on the topology of the system: adding the necessary links in the graph of excitations is incompatible with a bounded maximum degree.
However, connected components that are small compared to $n$ (in the sense of section~\ref{sec:characterization}) will always satisfy the global constraints.

In the approach of the previous section (where all constraints are assumed local), $\omega_0$ is chosen as a minimal set of excitations such that~\eqref{eq:close_S} is a syndrome. 
Suppose that there are global constraints and $\omega_0$ is chosen by taking into account only local ones.
The resulting $\omega'=\omega+\omega_0$ might not be a syndrome, but only with a probability that decreases exponentially in the system size.
Such `non-syndrome' events signal a large measurement error, and the simplest strategy in practice is to discard the computation.

\section{Gauge confinement}\label{sec:gauge_confinement}

This section discusses a mechanism that allows a topological stabilizer subsystem code to exhibit single-shot error correction.
Unlike in the previous section, here the relevant noise model is local in a standard way, as represented by $\cchan N_{\tau,\epsilon}^\mathrm{loc}$.
The formulation will be rather abstract, as a specific example will only be introduced in the next section.

\subsection {Gauge syndrome and errors}\label{sec:gauge_and_errors}

The starting point is some family of stabilizer subsystem codes with \emph{local gauge generators}, \emph{i.e.} every qubit is in the support of a bounded number of generators and every generator has support on a bounded number of qubits.
Unlike in the conventional case of local codes, however, this locality is a priori not enough to guarantee that ideal error syndrome extraction is quantum-local.
Moreover, here the gauge syndrome extraction will be relevant, not just the error syndrome extraction.
Thus, it is necessary to make an additional assumption, namely that there exists a suitable ordering of the local gauge generators $\gauge_0$ such that (i) measuring the elements of $\gauge_0$ in that order provides a gauge syndrome, (ii) the corresponding error syndrome is always correct (for ideal measurements) and (iii) the ordering is such that this gauge syndrome extraction is a local procedure, \emph {i.e.} has finite depth.
These requirements can easily be met for CSS codes with local gauge generators, such as 3D gauge color codes: it suffices to measure first all $X$-type generators and then all $Z$-type generators.
In fact, for CSS codes the correction of $X$ and $Z$ errors can be carried over for simplicity separately, using two separate gauge/error syndromes.
For other interesting codes~\cite{bombin:2010:subsystem} it is not clear if a suitable ordering exists, even though there exist local sequences of measurements of the elements of $\gauge_0$ (with some repeated elements) that do provide a correct \emph{error} syndrome.

As in section~\ref{sec:characterization}, a notion of connectivity for qubits and generators will be crucial.
Here and elsewhere $\sqcup$ denotes the union of disjoint sets.
Recall the notation $\err$ introduced in~\eqref{eq:err}.
\begin{defn}\label{defn:locally_connected}
A family of stabilizer subsystem codes is {\bf \emph{locally connected}} if it has local gauge generators and there exists a family of graphs of bounded maximum degree, one per code, satisfying the following. For each graph (i) the node set is $\gauge_0\sqcup \mathcal Q$, where $\mathcal Q$ is the set of qubits, and (ii) for any two mutually disconnected sets $\gamma_1\sqcup Q_1$ and $\gamma_2\sqcup Q_2$, $\gamma_i\subseteq \gauge_0$, $Q_i\subseteq \mathcal Q$, and for any two Pauli operators $E_1,E_2$ with $\supp E_i=Q_i$, $\gamma_1$ is a gauge syndrome if $\gamma_1\sqcup\gamma_2$ is a gauge syndrome and
\begin{equation}\label{eq:Egamma}
\synd (E_1E_2) = \err (\gamma_1\sqcup\gamma_2)\quad\Longrightarrow\quad \synd E_1 = \err\gamma_1.
\end{equation}
\end{defn}
The next section shows that the following property is enough to achieve single-shot error correction, given that the code family has an error correction threshold and admits a local gauge syndrome extraction procedure.
\begin{defn}
A family of stabilizer subsystem codes is {\bf\emph{$K$-confining}} for some $K>0$ if it is locally connected and for every gauge syndrome $\gamma$ there exists some Pauli operator $E$ with
\begin{equation}\label{eq:Qgamma}
\synd E=\err \gamma,\qquad |\supp E|\leq K|\gamma|.
\end{equation}
\end{defn}

\subsection {Single-shot error correction}\label{sec:gauge-ss}

Given a family of stabilizer subsystem codes as described in the previous section, quantum-local error correction proceeds as follows.
First the gauge generators $\gauge_0$ are measured, in the order prescribed, to recover a gauge syndrome.
To model noise in these measurements, assume that instead of the correct set $\gamma\subseteq\gauge_0$, defined as in section~\ref{sec:stabilizer}, the measurement yields a set $\gamma+\delta$ with a probability $r(\delta)$ independent of $\gamma$.
As discussed in section~\ref{sec:stabilizer}, some $\delta_0\subseteq\gauge_0$ must be chosen so that \eqref{eq:close_G} is valid, and this can be done in such a way that $\delta_0$ only depends on $\delta$.
Here we assume that $\delta_0$ is chosen to have minimal cardinality.
Regarding this choice, the comments of section~\ref{sec:self-ss} apply equally here.
Finally, the correction $\corr {\err (\gamma+\delta+\delta_0)}$ is applied.

A noisy recovery operation $\chan R$  is defined by a distribution $q_{\chan R}$ of wrong effective syndromes as in~\eqref{eq:noisy_R}.
In the procedure just discussed the effective wrong syndrome is $\err (\delta+\delta_0)$.
This leads to the following definition for the class of recovery operations, where the superscript $\mathrm{g}$ indicates that the generators $\gauge_0$ are measured directly.
\begin{defn}
The class $\cchan R_\eta^\mathrm{g}$ contains those recovery operations $\chan R$ such that
\begin{equation}
q_{\chan R}(\omega)=\sum_{\omega|\err(\delta+\delta_0)=\omega} r(\delta)
\end{equation}
for some probability distribution $r(\delta)$ that is $\eta$-bounded.
\end{defn}

\begin{thm}\label{thm:gauge}
For any $K$-confining family of stabilizer subsystem codes with failure probability bounded as in~\eqref{eq:f},
there exists some $\eta_0>0$ such that 
\begin{equation}\label{eq:goalloc}
\cchan R_{\eta}^\mathrm{g}\circ \cchan N_{\tau,\epsilon}^\mathrm{loc}\circ \chan C\subseteq \cchan N_{\tau',\epsilon+\delta}^\mathrm{loc}\circ \chan C
\end{equation}
holds for $\eta<\eta_0$ and every $\epsilon, \tau$, with $\delta$ as in~\eqref{eq:delta_taus} and
\begin{equation}\label{eq:taup_cc}
\tau' =  \frac {(\eta/\eta_0)^{k}}{1-(\eta/\eta_0)^{k}},\qquad k:=\frac 1 {2(1+K)}.
\end{equation}
\end{thm}
This result implies that $K$-confining code families with a local gauge extraction, as described in section~\ref{sec:gauge_and_errors}, allow single-shot error correction as long as they have a threshold for ideal error correction.
In particular, for local codes~\eqref{eq:local_codes} holds and the right properties are met, \emph{i.e.} (i) $\tau'$ can be made arbitrarily small just by reducing $\eta$ and  (ii) $\delta$ can be made arbitrarily small just by increasing the system size as long as $\tau+\tau'<\lambda_0$.
It is also worth remarking that unlike in section~\ref{sec:self-ss}, where the noise model had to be tailored for each specific code, here the standard local noise model $\cchan N_{\tau,\epsilon}^\mathrm{loc}$ of section~\ref{sec:local_noise} is enough.

\begin{proof}
According to lemma~\ref{lem:central} it suffices to show that for every $\chan R\in \cchan R_{\eta}^\mathrm{g}$
\begin{equation}\label{eq:goalRg}
\echan R\in \cchan N_{\tau',0}^\mathrm{loc},
\end{equation}
\emph{i.e.} that there exists a channel $\chan D$ with 
\begin{equation}\label{eq:condRDL}
\rechan R=\rchan D,\qquad \chan D\in\cchan L_{\tau'}.
\end{equation}
Choose for every gauge syndrome $\gamma$ a Pauli operator $E_\gamma$ with minimal support and such that
$\synd E_\gamma=\err \gamma$.
The first condition in~\eqref{eq:condRDL} holds setting
\begin{equation}
\chan D:= \sset{r(\delta)E_{\delta+\delta_0}}_{\delta}.
\end{equation}
It suffices to show that the distribution $p_{\chan D}'$ over the support of $E_{\delta+\delta_0}$, defined according to~\eqref{eq:local}, is $\tau'$-bounded.
This follows by applying lemma~\ref{lem:bounded} to each of the graphs that assure that the family of codes is locally connected.
Namely, take $\Gamma=\Gamma_1\cup\Gamma_2$, $\Gamma_1=\gauge_0$, $\Gamma_2=\mathcal Q$ and 
\begin{equation}
f(\delta)=(\delta+\delta_0)\cup\supp E_{\delta+\delta_0}.
\end{equation}
In this case \eqref{eq:connectivity} reads
\begin{equation}\label{eq:condconn}
2(1+K)|\delta\cap\delta_\mathrm c|\geq |\delta_\mathrm c|+|Q_\mathrm c|,
\end{equation}
where $\delta_\mathrm c\sqcup Q_\mathrm c$ is any connected component of $f(\delta)$, $\delta_\mathrm c\subseteq\gauge_0$, $Q_\mathrm c\subseteq \mathcal Q$.
By definition of local connectedness $\delta_{\text c}$ is a gauge syndrome and by the minimality of $\delta_0$, using the same argument as in~\eqref{eq:byminimality},
\begin{equation}\label{eq:condconn1}
2|\delta\cap \delta_\mathrm c|\geq |\delta_{\text c}|.
\end{equation}
If $E_\mathrm c$ is the restriction of $E_{\delta+\delta_0}$ to $Q_\mathrm c$, by definition of local connectedness $\synd E_\mathrm c=\err \delta_\mathrm c$.
For any Pauli operator $E'_\mathrm c$ such that $\synd E'_\mathrm c=\err\delta_\mathrm c$,
\begin{equation}
\synd E'=\synd E_{\delta+\delta_0},\qquad E':=E_{\delta+\delta_0} E_\mathrm c E'_\mathrm c,
\end{equation}
and
\begin{multline}
\supp (E_{\delta+\delta_0}E_\mathrm c)+|\supp E'_\mathrm c| \geq |\supp E'| \geq \\
\geq |\supp E_{\delta+\delta_0}|=|\supp (E_{\delta+\delta_0}E_\mathrm c)|+|Q_\mathrm c|.
\end{multline}
where  the second inequality is by the minimality of $E_{\delta+\delta_0}$.
Since the family is $K$-confined it is possible to choose $E'_\mathrm c$ such that 
\begin{equation}\label{eq:condconn2}
K|\delta_{\text c}|\geq |\supp E'_\mathrm c|\geq |Q_\mathrm c|.
\end{equation}
The inequalities~\eqref{eq:condconn1} and~\eqref{eq:condconn2} imply~\eqref{eq:condconn}.
\end{proof}

\subsection{Non-local constraints}\label{sec:global_constraints_gauge}

The discussion of section~\ref{sec:global_constraints} applies to a big extent also here, but there are interesting distinctions.
Here again the linear constraints that subsets of $\gauge_0$ have to satisfy in order to be gauge syndromes might be arranged in two types: local and non-local.
To clarify the distinction, recall that for each code in the family there is a graph representing the connectivity of gauge generators (and qubits), and that such graphs have a bounded maximum degree.
Local constraints are those that are represented in these graphs: if some $\gamma\subseteq\gauge_0$ satisfies them, so does each connected component of $\gamma$.
Non-local constraints, as in section~\ref{sec:global_constraints}, will then only be a problem if they are likely to occur.

A new scenario appears however due to the use of the gauge syndrome as a proxy for the error syndrome.
It might be the case that for all the subsets $\gamma\subseteq\gauge_0$ that satisfy the local constraints it is possible to extract an error syndrome.
Such sets $\gamma$, which will be termed \emph{gauge quasi-syndromes}, are not gauge syndromes but are equally useful.
To be precise, if the function $\errx$ can be extended to quasi-syndromes while preserving linearity, then all the results of the present section hold under the substitution of the words `gauge syndrome' by `gauge quasi-syndrome'.
This will be handy when dealing with 3D gauge color codes in the next section.

\section{Gauge color codes}

This section explains how single-shot error correction is possible in 3D gauge color codes~\cite{bombin:2013:gauge}. 
The emphasis in on the charge confinement mechanism: it provides the properties required for theorem~\ref{thm:gauge} to hold.
The section also outlines the procedure to obtain a universal set of quantum-local operations (in particular, local in 4D).

\subsection{Geometry}

3D gauge color codes are defined on lattices called 3-colexes~\cite{bombin:2007:branyons} in which $d$-cells are labeled with $d$ colors chosen from a set of $4$ colors: red (r), green (g), blue (b) and yellow (y).
For a given set of colors $\kappa$, a $\kappa$-cell is a cell with label $\kappa$.
If the lattice forms a closed 3-manifold the defining feature of 3-colexes is that each vertex is part of exactly one $d$-cell for each choice of $d$ colors.
Thus at a vertex 4 edges (1-cells), 6 plaquettes (2-cells) and 4 cells (3-cells) meet.
An interesting feature of 3-colexes is that their structure is fixed by giving only their colored 1-skeleton, \emph{i.e.} the graph formed by vertices and edges, together with the labels of the edges~\cite{bombin:2007:branyons}.

It is convenient to consider 3-colexes with colored boundaries. 
In particular, the boundary is a closed 2-manifold divided in \emph{regions}, with \emph{borders} between regions and \emph{corners} where the borders meet such that
\begin{itemize}
\item
each region is labeled with a color,
\item
two regions of different color meet at each border, which is labeled with these two colors, and
\item
three regions (and three borders) meet at each corner, which is labeled with these three colors.
\end{itemize}
The regions must match the lattice structure, so that regions are collections of plaquettes, borders are collections of edges, and corners are vertices.
For any vertex $v$, let $\kappa_v$ be the set of colors of the regions of which $v$ is part, in particular $\kappa_v=\emptyset$ for bulk vertices.
The defining feature for a 3-colex with boundary is that \emph{each vertex $v$ is part of exactly one $\kappa$-cell for each set of colors $\kappa$ such that
\begin{equation}\label{eq:defining}
|\kappa\cup\kappa_v|\neq 4.
\end{equation}
}
It follows that only vertices at corners are missing one of their edges, only edges at borders are missing one of their plaquettes, and only plaquettes at regions are missing one of their cells.
Moreover, the plaquettes forming a region have labels not containing the color of the region, and similarly for borders and edges.

\begin{defn}
A border is {\bf\emph{odd}} if it connects two corners with the same label.
\end{defn}
\begin{defn}
A region is {\bf\emph{free}} if it has an odd number of odd borders of any given color, and {\bf\emph{frozen}} if it has no odd borders.
\end{defn}
\begin{defn}
A 3-colex is {\bf\emph{simple}} if it is a topological ball and all its regions are topological disks and either free or frozen.
A 3-colex is free (frozen) if it is simple and all its regions are free (frozen).
\end{defn}
A simple 3-colex is frozen if and only if its regions only have 3 different colors.
The simplest geometry that provides a non-trivial free 3-colex involves 4 triangular regions, one of each color. These tetrahedral 3-colexes were introduced in~\cite{bombin:2007:3DCC}.

Explicit examples of 3D color codes with and without boundaries, including pictures, can be found in the literature~\cite{bombin:2007:branyons,bombin:2007:3DCC,bombin:2013:gauge,kubica:2015:universal,brown:2015:fault}.

\subsection{Codes}

To construct a 3D gauge color code from a 3-colex, first attach a qubit to each vertex.
Denote by $X_o$ the operator that flips the qubits in the set or geometrical object $o$, and similarly for $Z_S$.
The gauge group has generators
\begin{equation}\label{eq:defgauge}
\gauge_0:=\sset{X_p, Z_p}_{p\in \text{$2$-cells}}
\end{equation}
The definition \eqref{eq:defgauge} is enough to define the codes up to the choice of signs for the stabilizer generators.
Notice that 3D gauge color codes are self-dual CSS codes.

Given a region $R$, it is easy to check that the region operators $X_R$ and $Z_R$ belong to $\cent\gauge$.
Thus they are either nontrivial bare logical operators or stabilizers.
Their commutation rules are as follows.
\begin{prop}\label{prop:regionop}
Given a 3-colex and two different regions $R$ and $R'$ (i) $X_R$ and $Z_{R'}$ anticommute if and only if $R$ and $R'$ share and odd number of odd borders, and (ii) $X_R$ and $Z_{R}$ anticommute if and only if $R$ is free.
\end{prop}

\begin{sketchproof}
It is easy to check that a border is odd if and only if it contains an odd number of vertices.
Then (i) follows noting that $R$ and $R'$ can only share borders that are disjoint from each other.
As for (ii), let $\kappa$ be a set of colors containing the color of $R$ and let $\bar\kappa$ be the set of complementary colors to $\kappa$.
It is easy to check that
\begin{equation}
X_R=\prod_{p\in \text{$\bar\kappa$-cells in $R$}} X_p\,\,\prod_{b \in \text {$\kappa$-borders of $R$}} X_b.
\end{equation}
Since the product of plaquettes is in $\gauge$, the result follows.
\end{sketchproof}

For closed 3-colexes with the topology of a 3-sphere, it follows from the properties of conventional 3D color codes~\cite{bombin:2007:branyons} that (i) the stabilizer has generating set
\begin{equation}\label{eq:locstab}
\stab_0:=\sset{X_c, Z_c}_{c\in \text{$3$-cells}},
\end{equation}
and (ii) there are no encoded qubits.
Simple 3-colexes are more interesting.

\begin{prop}\label{prop:stablog}
A 3D gauge color code defined on a simple 3-colex admits the stabilizer generating set
\begin{equation}\label{eq:stabsimple}
\sset{X_c, Z_c}_{c\in \text{$3$-cells}}\sqcup\sset{X_R, Z_R}_{R\in \text{frozen regions}},
\end{equation}
and bare logical operators are, up to stabilizers and phases, generated by the set
\begin{equation}\label{eq:centsimple}
\sset{X_R, Z_R}_{R\in \text{free regions}}.
\end{equation}
\end{prop}

\begin{sketchproof}
The key is the following geometrical construction.
A simple 3-colexes can be glued to a duplicate of itself to produce a closed 3-colex with the topology of a 3-sphere.
Namely, it suffices to join in the obvious way corresponding corners with an edge, corresponding borders with a plaquette and corresponding regions with a cell.
The coloring of the cells joining the two copies of the original 3-colex is the complementary of the set of colors labeling the regions, borders or corners that the cell joins.

Let $\gauge$ and $\gauge'$ be the gauge groups for the original 3-colex and its duplicate, and $\bar\gauge$ be the gauge group for the closed 3-colex.
The generators of $\bar\gauge$ include the generators of $\gauge$ and $\gauge'$ plus those corresponding to the plaquettes in the \emph{interface} between the two copies.
Given a Pauli operator $E$ on the first copy, let $E'$ be the analogous operator acting on the second copy.
Interface plaquettes are completely symmetrical up to the exchange of the copies, which implies that the product $EE'$ commutes with $X_p$ or $Z_p$ for any interface plaquette $p$.
It follows that
\begin{multline}\label{eq:dup}
E\in\cent{\gauge}\iff E E'\in\cent{\bar\gauge}\iff\\\iff EE'\in \bar\stab\iff E\in \stab^\star,
\end{multline}
where in the second line containment is up to a phase and the elements of $\stab^\star$ are the restriction to the original 3-colex of the elements of the stabilizer $\bar\stab$ of the closed 3-manifold (recall that $\bar\stab\propto\cent{\bar\gauge}$).
Since the vertices of an interface cell in the original colex are just the vertices of the corresponding region,
it follows from~\eqref{eq:dup} that, up to phases, $\cent{\gauge}$ has generating set 
\begin{equation}\label{eq:stabsimple}
\sset{X_c, Z_c}_{c\in \text{$3$-cells}}\sqcup\sset{X_R, Z_R}_{R\in \text{frozen regions}},
\end{equation}
The result follows using proposition~\ref{prop:regionop}.
\end{sketchproof}

The stabilizer generator $X_R$ of a frozen region $R$ is the product of the gauge generators $X_p$ of some of the plaquettes $p$ in the region.
Namely, if $R$ is a b-region that borders with g- and r-regions,
\begin{equation}\label{eq:Rdecomp}
X_R=\prod_{p\in \text{rg-cells in $R$}} X_p,
\end{equation}
and similarly for other color combinations and for $Z_R$.
Indeed, (i) all vertices belong to exactly one rg-cell unless they belong to a by-border, and no by-borders share vertices with such a region, and (ii) if a vertex belongs to the region so does its rg-plaquette, because it cannot belong to the neighboring regions. 
It is thus important that the plaquettes in~\eqref{eq:Rdecomp} are rg-cells, no other color choice is possible.
This is in contrast with a similar relation for cells.
Namely, if $c$ is a b-cell~\cite{bombin:2007:branyons}
\begin{equation}\label{eq:cdecomp}
X_c=\prod_{p\in \text{$\kappa$-cells in $c$}} X_p,\qquad \kappa\in\sset{\text{rg, ry, gy}}.
\end{equation}

According to proposition~\ref{prop:stablog} (i) for a free simple colex the stabilizer is generated by cell operators, and (ii) for a frozen simple colex the stabilizer is generated by cell and region operators.
In the latter case not all region operators are needed to form a set of independent generators of the stabilizer.
\emph{E.g.} if the frozen colex has only rgb-corners the following relation holds
\begin{equation}\label{eq:redundant}
\prod_{c \in\text{rgb-cells}} X_c \prod_{c \in\text{rgy-cells}} X_c =\prod_{R \in\text{b-regions}} X_c,
\end{equation}
and similarly for r- and g-regions.
Indeed, each vertex belongs (i) exactly to one rgb-cell and (ii) exactly to one rgy-cell, \emph{unless} it belongs to a b-region.

\subsection{Dual picture}

When discussing error correction, and also when constructing codes~\cite{bombin:2013:gauge}, it is better to use the dual lattice rather than the 3-colex.
For a closed 3-colex the dual is a simplicial lattice with 4-colored vertices.
In particular, the color of a dual vertex is the complementary to the three colors labeling the corresponding cell.
For the sake of error correction dual edges should be labeled with the colors of the corresponding plaquettes, \emph{e.g.} a yb-edge connects a r-vertex and a g-vertex.
Only edges and vertices are relevant for error correction.

\emph{From this section onwards only simple 3-colexes are considered.}
For a simple 3-colex with a boundary a vertex must be attached to each region, with the same color as the region.
Vertices that are dual to cells of the 3-colex are \emph{internal}, and the rest are \emph{external}.
External vertices corresponding to free (frozen) regions are free (frozen), and internal vertices are frozen.

For explicit graphical examples of the duality, see~\cite{bombin:2013:gauge}.
In what follows only the dual picture will be used, and thus the wording `dual' will be dropped.

\subsection{Syndromes}

As discussed in section~\ref{sec:gauge_and_errors}, for CSS codes two separate recovery operations can be considered for bit-flip and phase-flip errors.
3D gauge color codes are in addition self-dual, so that the two recovery operations are analogous.
Thus it suffices to discuss the detection and correction of bit-flip errors, in which $Z$-type stabilizer and gauge generators are involved.

The $Z$-type generators are as follows.
There is a stabilizer generator $s_v$ per frozen vertex $v$, and thus an error syndrome $\sigma$ can and will be identified with a set of frozen vertices.
Similarly there is a gauge generator $g_l$ per edge $l$, and a gauge syndrome $\gamma$ will be identified with a set of edges.

As discussed in section~\ref{sec:stabilizer}, there exist linear constraints that characterize which sets of vertices/edges are error/gauge syndromes.
These constraints have a geometrical nature~\cite{bombin:2007:branyons}.
The most important amounts to the conservation of a certain `color flux', and has its origin in~\eqref{eq:cdecomp}.
Let an edge carry a color $\kappa$ if one of its two color labels is $\kappa$.
Given a set of edges $\gamma$, consider its subsets
\begin{equation}\label{eq:gammakappa}
\gamma_\kappa:=\set{l\in \gamma}{\text{$l$ carries $\kappa$}}.
\end{equation}
If $\gamma$ is a gauge syndrome then every $\gamma_\kappa$ is closed, \emph{i.e.} at every internal vertex an even number of edges in $\gamma_\kappa$ meet.
At an external vertex $\gamma_\kappa$ might have an endpoint, which should be visualized as the flux ending at the boundary region.
This is only possible if it is not a $\kappa$-region, because edges with an endpoint at the external $\kappa$-vertex cannot carry color $\kappa$.
Because these are the only local constraints on gauge syndromes, a set of edges $\gamma$ is a gauge quasi-syndrome if it satisfies color-flux conservation,~\emph{i.e.} if every $\gamma_\kappa$ is closed, see section~\ref{sec:global_constraints_gauge}.

The geometrical relationship between a gauge syndrome $\gamma$ and its stabilizer syndrome $\sigma=\err\gamma$ follows from (\ref{eq:Rdecomp},\ref{eq:cdecomp}).
Consider an internal r-vertex $v$.
The edges meeting at $v$ are gb-, by- and gy-edges.
Among the gb-edges some belong to $\gamma$.
Consider this set of edges, together with the analogous ones with by- and gy-edges.
For color flux to be conserved there are two possibilities, either (i) the three sets are even or (ii) the three sets are odd.
In the second case $v$ is a \emph{branching point} of $\gamma$. 
The definition for internal g-, b- and y-vertices is entirely analogous.
The case of a frozen external r-vertex $v$ is different.
\emph{E.g.} let the region corresponding to $v$ have border with g- and b-regions.
Then $v$ is a \emph{termination point} of $\gamma$ if the number of gb-edges of $\gamma$ with an endpoint in $v$ is odd.
The syndrome $\sigma$ is composed of the branching and termination points of $\gamma$.
More generally, any gauge quasi-syndrome $\gamma$ has a well-defined set of branching and terminations points.
Indeed, the function $\errx$ can be extended to quasi-syndromes while preserving linearity, see section~\ref{sec:global_constraints_gauge}.
Notice however that the image of $\errx$ might now include sets of vertices that are not error syndromes.


\subsection{Gauss law}

According to the discussion in the previous section a gauge syndrome $\gamma$ can be visualized as a collection of two-colored `flux' strings that form a net and might have endpoints at boundary regions.
The branching and termination points $\sigma$ of this net have a color.
The subject of this section is the significance of color and its connection with flux.

For an encoded state the error syndrome is $\sigma=\emptyset$.
This is the `vacuum'.
Applying local bit-flip operations will change $\sigma$ in a well-defined way, adding vertices to it.
These are `particles'.
An arbitrary set of particles cannot be created locally, something that can be formulated stating that the elements of $\sigma$ carry a $\mathbf Z_2\times\mathbf Z_2\times\mathbf Z_2$ charge~\cite{bombin:2007:branyons}.
To specify how color labels fit into this group, construct it as the abelian group with generators r, g, b and y subject only to the relations, in additive notation,
\begin{equation}
\text r+\text g+\text b+\text y=2\text r=2\text g=2\text b=0.
\end{equation}
Only neutral sets of charges (particles), \emph{i. e.} such that their sum is zero, can be created locally. And conversely, any neutral set of charges can be created locally.

In fact, this is only true in the bulk of the lattice.
In the vicinity of a free boundary region of color $\kappa$ it is possible to create charges with total charge $\kappa$~\cite{bombin:2007:3DCC}.
This can be visualized as an `immaterial' charge $\kappa$ being created at the same time on the free external $\kappa$-vertex.

Charges $\sigma$ (error syndromes) and fluxes $\gamma$ (gauge syndromes) satisfy a Gauss law.
Namely, for a `volume' within the bulk there is a direct correspondence between the charge contained in the volume and the parity of the number of flux strings with a given coloring crossing its boundary.
In particular
\begin{itemize}
\item For neutral charge all the parities are even.
\item For total charge r the parities for rg-, rb- and ry-flux strings is even, and the rest odd.
\item For total charge r+g the parities for rg- and by-flux strings is even, and the rest odd.
\end{itemize}
The rest of cases are analogous and need not be listed.
Notice that there is a one-to-one correspondence between charge and flux configuration.

Instead of formulating this property for the whole volume, it is equally possible to apply it to each separate connected component of $\gamma$.
In fact, this not only applies to gauge syndromes but also to any gauge quasi-syndrome $\gamma$.
An immediate consequence of this is that every connected component of a gauge quasi-syndrome has branching points with neutral charge.
More exactly, this is true as long as $\gamma$ does not end on a boundary region, \emph{i.e.} it does not contain an edge connected to an external vertex.
More generally the following result holds.

\begin{prop}\label{prop:neutralization}
Let $\gamma$ be a gauge quasi-syndrome that ends at boundary regions $\sset {B_i}$.
Let $c$ be the total charge of its branching and termination points.
If the 3-colex is free $c$ is the sum of the colors of some of the regions $B_i$, and if the 3-colex is frozen $c$ is neutral.
\end{prop}
\begin{sketchproof}
Suppose that the 3-colex is free.
The property can be checked in a case by case basis.
\emph{E.g.} if the total charge is r there must be at least a gb-, a gy- and a by-flux ending on boundary regions, and each of these regions can have two possible colors, but all the combinations are compatible with a total charge r.

Now suppose that the 3-colex is frozen. 
Without lost of generality, suppose that the regions have colors r, g and b.
To check this property one can disregard ry-, gy- and by-fluxes that end at the boundary.
The rest have a termination point with a definite charge, \emph{e.g} a rg-flux can only end at b-region.
Again the property can be checked in a case by case basis, \emph{e.g.} if the total charge of the branching points is $y$ there is an odd number of rg-, gb- and rb-colored termination points and thus the total charge is neutral.
\end{sketchproof}

\emph{E.g} in the free case if a connected component of $\gamma$ is connected to a b-region and to a y-region its total charge could be neutral, y, b or y+b. 
It is the confinement of charge that makes single-shot error correction possible.
The above proposition is closely connected to the fact that non-trivial charge can be confined, except in the vicinity of a free boundary region, where charges that can disappear in the boundary are not confined.

Free 3-colexes and frozen 3-colexes seem to be the most important for applications~\cite{bombin:2014:dimensional}.
The above proposition shows that in such simple 3-colexes every gauge quasi-syndrome $\gamma$ gives rise to an error syndrome $\sigma$.
Thus the redefined function $\errx$ maps gauge quasi-syndromes to error syndromes, and this is enough in the sense of section~\ref{sec:global_constraints_gauge}.
In what follows only free and frozen 3-colexes are considered.
For this reason the distinction between gauge quasi-syndromes and gauge-syndromes is dropped, see section~\ref{sec:global_constraints_gauge}.

\subsection {Faulty gauge syndromes}\label{sec:faulty_gauge_syndrome}

The notion of charge is also useful to understand the classical error correction of faulty gauge syndromes.
As discussed in the preceding sections, a collection of edges $\gamma$ is a gauge syndrome (strictly speaking a quasi-syndrome, see the preceding section) if all the $\gamma_\kappa$ of~\eqref{eq:gammakappa} are closed.
To characterize how $\gamma$ fails to satisfy this linear conditions it is convenient to introduce a linear operator $\tilde\partial$ mapping sets of edges to a syndrome.
In particular, this is the syndrome used to correct the errors in the measurements of gauge generators.

The idea is that $\tilde\partial$ is a boundary operator that takes color into account: in a sense, it can see the $\gamma_\kappa$ structure within $\gamma$.
In particular each syndrome $\tilde\partial\gamma$ is a function mapping internal vertices to the group $\mathbf Z_2\times\mathbf Z_2\times\mathbf Z_2$: it labels internal vertices with a charge.
This charge should not be confused with the one defined for error syndromes in the previous section.
Even though the two charge groups are isomorphic, this time the abstract generators are the labels rg, rb, ry, gb, gy and by, subject to the relations
\begin{align}
\text {gb}+\text {by}+\text {gy}=\text {rb}+\text {by}+\text {ry}=\text {rg}+\text {gy}+\text {ry}=0,\\
2\text {rg}=2\text {gb}=2\text {by}=0.
\end{align}
By definition, the charge $\tilde\partial\gamma(v)$ at an internal vertex $v$ is the sum of the labels of the edges in $\gamma$ meeting at $v$.
Notice that $\gamma$ is a gauge syndrome if and only if all internal vertices have trivial charge, \emph{i.e.} if $\tilde\partial\gamma=0$.

Every edge creates equal charges on both of its endpoints, and $\tilde\partial$ is a linear function. 
This implies that if $\gamma$ only has edges in the bulk of the lattice then its total charge is neutral, \emph{i.e.} $\sum_v \tilde\partial\gamma(v)=0$. And conversely, the structure of the lattice guarantees that there is always some local set of edges for a given neutral charge configuration.

With boundaries the situation is again different.
In the vicinity of a $\kappa$-region syndromes need not be neutral. 
Instead the total charge might be any sum of labels not including the color $\kappa$.
\emph{E.g.} in the vicinity of a r-boundary the syndrome of a local set of edges $\gamma$ might have total charge 0, gb, gy or by.
The reason is that there are gb-, gy- and by-edges that connect some internal vertex with an external r-vertex, and no charge is assigned to external vertices.

\subsection {Single-shot error correction}\label{sec:ss-gcc}

Since $X$ and $Z$ errors are treated separately and are analogous, as usual it suffices to consider bit-flip errors.
Assume that we are given a family of simple 3-colexes, either free or frozen, all with the same local structure, translationally invariant up to boundary defects.
The corresponding 3D gauge color codes have local (Z-type) gauge generators $g_l$ attached to edges $l$.
Assume in addition that the ($Z$-type) stabilizer has \emph{local} generators
\begin{equation}
\stab_0:=\set {s_v}{v\in V_0}
\end{equation}
where (i) $V_0$ is some set of frozen vertices that includes all internal vertices, (ii) if $v$ is an internal vertex $s_v=Z_c$ with $c$ the corresponding 3-colex cell, and (iii) if $v$ is an external vertex $s_v=Z_R$ with $R$ the corresponding 3-colex region (which must be small).

This is a family of local codes and thus it exhibits a threshold for ideal error correction under local noise.
Theorem~\ref{thm:gauge} implies then that single-shot error correction is possible if $K$-confinement is satisfied.
The purpose of this section is to show that this is indeed the case.

According to definition~\ref{defn:locally_connected}, for the family of codes to be locally connected a suitable family of graphs has to exist.
Each graph has as nodes the qubits (or vertices of the 3-colex) and the gauge generators (or edges of the dual lattice).
To define the graph, let two nodes of the graph be linked if and only if they share a vertex $v\in V_0$, given that:
\begin{itemize}
\item
$v$ is a vertex of a qubit $q$ if $q\in\supp s_v$, and 
\item
$v$ is a vertex of an edge $l$ if $v$ is an endpoint of $l$.
\end{itemize}

Condition (i) in definition~\ref{defn:locally_connected} is satisfied.
To check condition (ii) consider Pauli operators $E_i$ and sets $\gamma_i\subseteq\gauge_0$, $Q_i\subseteq\mathcal Q$, subject to the conditions specified in the definition.
First, since $\gamma_1$ and $\gamma_2$ are disconnected they do not share any internal vertex and thus
\begin{equation}
\quad\tilde\partial\gamma_1\neq 0\quad\Longrightarrow\quad\tilde\partial(\gamma_1\sqcup\gamma_2)\neq 0.
\end{equation}
If $\gamma_1\sqcup\gamma_2$ is a gauge syndrome, so is $\gamma_1$.
Second, the sets $\err\gamma_i$ and $\synd E_i$ are subsets of $V_0$, and the definition is such that if two sets among the $\gamma_i$ and $Q_i$ are disconnected from each other then their $V_0$ subsets are disjoint. 
It is then easy to derive~\eqref{eq:Egamma}: the family of codes is locally connected.

For every gauge syndrome $\gamma$ there exists some bit-flip operator $E$ with 
\begin{equation}
\synd E=\err\gamma, \qquad \supp E\subseteq \bigcup_{v\in I_\gamma} \supp s_v,
\end{equation}
where $I_\gamma$ is the set of internal vertices of $\gamma$.
This is guaranteed by the string operator construction of~\cite{bombin:2007:branyons}, which can be applied thanks to the charge neutralization properties of proposition~\ref{prop:neutralization}.
If for the given family of codes $K$ is the maximum number of qubits in the support of a local stabilizer generator, then 
\begin{equation}
|\supp E| \leq K |I_\gamma|\leq K|\gamma|.
\end{equation}
The family of 3D gauge color codes is $K$-confined, as required.

\subsection{Quantum-local fault-tolerant operations}

For this section the relevant 3D gauge color codes are those introduced in~\cite{bombin:2013:gauge}, which encode a single qubit and are built from tetrahedron-shaped (and thus free) 3-colexes.

3D gauge color codes are CSS topological stabilizer codes, and as such the initialization of eigenstates of the logical $X$ and $Z$ operators is a quantum-local operation~\cite{dennis:2002:tqm}.
Measurements in the $X$ and $Z$ basis are also quantum-local.
In addition, 3D gauge color codes allow the quantum-local implementation of a universal set of gates via gauge fixing~\cite{bombin:2013:gauge}.
Moreover, all these quantum-local operations are actually local except for the error-correcting part that they incorporate, suggesting that quantum-locality could be preserved in a fault-tolerant setting.
This is indeed the case, as outlined next.

The initialization of a encoded eigenstate $|0\rangle$ of the logical $Z$ is performed by initializing all the physical qubits to $|0\rangle$ and then correcting phase-flip $Z$ errors~\cite{dennis:2002:tqm}.
The initialization of eigenstates of $X$ is analogous.
In both cases error correction is needed just to bring the state to the code subspace: logical errors do not play any role due to the nature of the encoded states.
Clearly single-shot error correction, as described above, will succeed at making the procedure fault-tolerant and provide encoded eigenstates up to some noisy channel in $\cchan N_{\tau,\epsilon}^\mathrm{loc}$ where, as usual, $\tau$ can be decreased by improving the accuracy of the operations and $\epsilon$ by increasing the system size.

The measurement of the logical Z amounts to measure each single-qubit $Z$ and then perform bit-flip error correction on the resulting classical bits before reading the logical value.
The measurement of the logical X is analogous.
In both cases measurement errors can be regarded as conventional single-qubit errors that happened before an ideal measurement.
Thus the whole process can be modeled as some noisy channel in $\cchan L_\lambda$ (with $\lambda$ representing the accuracy of the single-qubit measurements) followed by ideal error correction and the measurement of the corresponding operator.

Hadamard gates are transversal and as such can be simply modeled by a perfect gate followed by some noisy channel in $\cchan L_\lambda$.
Notice that the perfect Hadamard exchanges $X$ and $Z$ errors.
The case of CNot gates is similar, with the peculiarity that $X$ or $Z$ errors of the two codes are combined in one of them.
This combination is not problematic due to \eqref{eq:compN}, and the correlations induced are also not an issue because the circuit has bounded depth.

Completing the universal set of gates requires gauge fixing~\cite{bombin:2013:gauge}.
The gauge fixing operation amounts to (i) applying the single-shot error correction procedure to bit-flip errors, so that some bit-flip error $E$ is chosen for an effective $Z$-type syndrome $\gamma+\delta_\mathrm {eff}$ with
\begin{equation}
\synd E = \err (\gamma+\delta_\mathrm {eff})
\end{equation}
and (ii) applying a gauge bit-flip operator $G$ such that
\begin{equation}
\syndx_g G=\syndx_g E+\gamma+\delta_\mathrm {eff}
\end{equation}
where $\syndx_g A$ denotes the $Z$-type error syndrome of the conventional 3D color code~\cite{bombin:2007:3DCC}, which is the same as a $Z$-type gauge syndrome in the gauge code.
The purpose of the second step is to set the $Z$-type gauge syndrome to its trivial value.
If the correct gauge syndrome is $\gamma$, however, the net effect is instead to fix the gauge syndrome to $\delta_\mathrm{eff}$, which is confined for noise below a threshold.

Once the gauge fixing is completed a transversal gate can be applied that completes the universal gate set together with the Hadamard and CNot~\cite{bombin:2013:gauge}.
This gate only preserves the subspace of trivial syndrome, but in addition it leaves phase-flip errors invariant, and maps a set of bit-flip errors that has as syndrome a given connected net $\delta_\mathrm {eff}$ to a (non-Pauli) error that, from the gauge code perspective, can be chosen to have support within a close neighborhood of $\delta_\mathrm{eff}$.
Thus the residual noise after gauge fixing and the transversal gate have been performed shows the required confinement, \emph{i.e.} can be described by a suitable $\cchan N_{\tau',\epsilon+\delta}^\mathrm{loc}$.

In summary, 3D gauge color codes provide a somewhat straightforward fault-tolerant quantum computing scheme with a constant time overhead.
Or rather, there is no overhead if the classical computation time required for error correction is disregarded.
This might be enough for practical applications in the sense that all the required computations can be carried out efficiently, as discussed in section~\ref{sec:complexity}.
An important aspect is that all quantum operations are local, but only in four spatial dimensions.
Three of these are required by the the 3D codes, and the forth is for the CNOT gate, which involves pairs of codes that have to be next to each other without overlapping.
This difficulty can be overcome~\cite{bombin:2014:dimensional}.

\section{Efficient error correction}\label{sec:complexity}

A weakness of the threshold results \eqref{eq:local_codes}, \eqref{eq:taup_sc} and \eqref{eq:taup_cc} is that they rely on the ability to find errors of minimal support for a given syndrome.
If this task cannot be done efficiently, the methods are impractical and the thresholds possibly useless.

Fortunately, a closer inspection of the respective proofs reveals that in fact it is enough to be able to find errors that are only a constant away from optimal, connected component by connected component.
To make this explicit, consider in particular theorem~\ref{thm:gauge}.
It is enough to find a $\delta_0$ such that, for any $\delta$ with $\delta+\delta_0$ a gauge syndrome and for any connected component $\delta_{\text c}$ of $\delta+\delta_0$,
\begin{equation}\label{eq:delta1}
|\delta_0\cap \delta_{\text c}|\leq c |\delta\cap\delta_{\text c}|
\end{equation}
for some constant $c\geq 1$.
The price to pay is that the constant $k$ in \eqref{eq:taup_cc} will be worse, in particular
\begin{equation}
k=\frac 1 {(1+c)(1+K)}.
\end{equation}

Another aspect that admits changes is connectivity, as long as the resulting connected components have at most a size proportional to the subset of errors that they contain.
Continuing with the above example, suppose that for each $\delta$ there is a set $\mu$ with elements in some set $M$.
The elements of $M$ might be connected to each other and to the edges in the lattice (recall that the elements of $\delta$ are edges), but the connections must preserve the locality of the lattice.
Instead of taking connected components of the set $\delta+\delta_0$, one can then consider the set $(\delta+\delta_0)\sqcup\mu$ (or alternatively the set $(\delta\cup\delta_0)\sqcup\mu$, which works just as well).
If each resulting connected component takes the form $\delta_c\sqcup\mu_c$, it suffices for the proof technique to work that
\begin{equation}\label{eq:mu1}
|\mu_c|\leq a |\delta\cap\delta_{\text c}|
\end{equation}
for some constant $a$.
This will result in a different constant $\eta_0$ and, again, some extra factor in the value of $k$ in \eqref{eq:taup_cc}.

The purpose of this section is to use these observations to map the error correction problems that are relevant to 3D gauge color codes into problems that have an efficient solution but still provide a threshold.
Notice that the threshold results \eqref{eq:delta_sc} and \eqref{eq:taumu_sc} do not suffer from similar dificulties because they rely on error correction methods that are efficient, at least for the codes considered in \cite{bombin:2013:self}.

Ultimately the decoding methods discussed here are only of theoretical interest, as they are designed to facilitate proofs, rather than to be efficient.
For practical applications it is conceivable to use very different decoding approaches, even those for which it is not a priori clear that single-shot error correction should work.
The success of such strategies~\cite{brown:2015:fault} is an indication of the robustness of the concept.

\subsection{Strategy}

This section provides a general strategy to approximate, in the sense of \eqref{eq:delta1}, the optimal solution of a given error correction problem in terms of another problem that can be efficiently computed.

Let the problem of interest be defined by two abstract sets $E$ and $S$ and a linear map $\syndx$ that maps errors, \emph{i.e.} subsets $\epsilon\subseteq E$, to syndromes $\sigma\subseteq S$.
In addition there is a graph with node set $E$ that represents the connectivity of the problem.
Two elements of $e_1,e_2\in E$ are connected if their syndromes $\synd{\sset{e_i}}$ overlap (but there might be other connections).
Clearly, if $\epsilon_{\text c}$ is a connected component of $\epsilon\subseteq E$ then
\begin{equation}\label{eq:conepsilon}
\synd \epsilon=\emptyset \quad\Longrightarrow\quad \synd \epsilon_{\text c}=\emptyset.
\end{equation}
A set $\sigma\subseteq S$ is a syndrome it it has a solution, \emph{i.e.} some error $\epsilon$ with $\synd\epsilon=\sigma$. 
A minimal solution is one that has minimal cardinality.
An error $\xi\subseteq E$ is elementary if $|\xi|=1$.

Assume that there is a second problem $E',S', \syndx'$, with $S'=S$, together with
\begin{enumerate}
\item
a function $m$ that provides for a given syndrome $\sigma\subseteq S$ a minimal solution $m(\sigma)\subseteq E'$,
\item
a linear function $g$ mapping errors $\epsilon'\subseteq E'$ to errors $g(\epsilon')\subseteq E$ and such that
\begin{equation}\label{eq:circularity2}
\syndx \circ g = \syndx',
\end{equation}
\item
integers $a,b$ such that for any elementary errors $\xi\subseteq E$, $\xi'\subseteq E'$
\begin{equation}\label{eq:ab}
 |(m\circ \syndx )(\xi)|\leq a, \quad |g(\xi')|\leq b.
\end{equation}
\end{enumerate}
Consider now a graph with node set $E\sqcup E'$, containing the previous graphs for $E$ and $E'$ and the following new links: $e\in E$ is connected to $e'\in E'$ if their syndromes overlap or if $e\in g(\sset{e'})$.
If $\epsilon_{\text c}\sqcup\epsilon_{\text c}'$ is a connected component of $\epsilon\sqcup\epsilon'$ it is clear that
\begin{align}
\synd \epsilon=\synd'\epsilon'\quad&\Longrightarrow\quad \synd\epsilon_{\text c}=\synd'{\epsilon'_{\text c}},\label{eq:eps1}\\
\epsilon=g(\epsilon')\quad&\Longrightarrow\quad \synd\epsilon_{\text c}=g(\epsilon'_{\text c}).\label{eq:eps2}
\end{align}

The map $g$ provides a solution $\epsilon_0$ for a given syndrome $\sigma\subseteq S$ constructed via a minimal $\epsilon_0'\subseteq E'$, namely
\begin{equation}\label{eq:epsilon1}
\epsilon_0:=g(\epsilon_0'),\qquad \epsilon_0'=m(\sigma).
\end{equation}
\begin{prop} 
If for some $\epsilon\subseteq E$
\begin{equation}\label{eq:eqepsilon}
\synd \epsilon=\synd{\epsilon_0},
\end{equation}
and $\epsilon_{\text c}\sqcup\epsilon'_{\text c}$ is a connected component of $(\epsilon\cup\epsilon_0)\sqcup\epsilon_0'$, then
\begin{align}\label{eq:aprox_min}
|\epsilon_0\cap\epsilon_{\text c}|\leq ab|\epsilon\cap\epsilon_{\text c}|,\qquad
|\epsilon'_{\text c}|\leq a|\epsilon\cap\epsilon_{\text c}|.
\end{align}
\end{prop}
This is precisely what was required in \eqref{eq:delta1} and \eqref{eq:mu1}, under the obvious identifications (\emph{e.g.} $\epsilon'_0$ is $\mu$).

\begin{proof}
Let $h$ be the linear map taking errors $\epsilon\subseteq E$ to errors $h(\epsilon)\subseteq E'$ and such that for any elementary $\xi\subseteq E$ 
\begin{equation}
h(\xi)=(m \circ \synd)(\xi).
\end{equation}
By linearity and since $\syndx'\circ m$ is the identity
\begin{equation}
\syndx' \circ h = \syndx.
\end{equation}
In particular
\begin{equation}
\synd' {h(\epsilon\cap\epsilon_{\text c})} = \synd {(\epsilon\cap\epsilon_{\text c})}=\synd'{\epsilon_{\text c}'}.
\end{equation}
where the second identity uses \eqref{eq:eps1} applied to $\epsilon\sqcup\epsilon_0'$ (each of its connected components is a subset of a connected components of $(\epsilon\cup\epsilon_0)\sqcup\epsilon_0'$).
By the minimality of $|\epsilon_0'|$
\begin{equation}\label{eq:mfh}
|\epsilon_{\text c}'|\leq |h(\epsilon\cap\epsilon_{\text c})|.
\end{equation}
Notice that $|A+B|\leq |A|+|B|$ for any sets $A$, $B$.
Then by linearity \eqref{eq:ab} yields
\begin{equation}\label{eq:hg}
 |h(\epsilon\cap\epsilon_{\text c})|\leq a|\epsilon\cap\epsilon_{\text c}|, \quad |g(\epsilon_{\text c}')|\leq b|\epsilon_{\text c}'|
\end{equation}
which combined with \eqref{eq:eps2} applied to $\epsilon_0\sqcup\epsilon_0'$ and \eqref{eq:mfh} gives \eqref{eq:aprox_min}.
\end{proof}

\subsection{Efficient solution}\label{sec:efficient_solution}

Given an error correction problems $E,S,\syndx$ of interest, the goal is to map it, in the sense of the previous section, to some other problem $E',S,\syndx'$ with an $m$ that can be efficiently computed.
There is a problem with an efficient solution that appeared early on in the study of topological stabilizer codes. 
$S$ is the set of nodes of a graph, $E'$ the set of edges, and $\syndx'$ is the boundary operator $\partial$ mapping each edge to its two endpoints.
For graphs in which paths of minimal length between two points can be computed efficiently (\emph{e.g.} in a uniform lattice embedded in $\mathbf R^n$), $m$ is efficient via minimum weight matching~\cite{dennis:2002:tqm}.
Moreover, a slight modification allows to introduce in $E'$ also `edges' with a single endpoint, so that $\partial$ maps such an edge to its endpoint.

An error correction problem $E,S,\syndx$ can be mapped to such $E',S,\syndx'$ if there exists integers $a$ and $b$ such that for every elementary error $\xi\subseteq E$ there exists errors $\epsilon_i\subseteq E$, $i=1,\dots, r$, with
\begin{equation}\label{eq:epsilon_i}
\synd \xi=\sum_{i=1}^r \synd {\epsilon_{i}}, \quad 1\leq |\synd{\epsilon_{i}}|\leq 2, \quad r\leq a, \quad |\epsilon_{i}|\leq b.
\end{equation}
The construction is as follows.
For every pair $\xi, i$ the set $E'$ contains an edge $e$ with endpoint(s) the element(s) of $\synd {\epsilon_{i}}$, and 
\begin{equation}
g(\sset{e})=\epsilon_{i}.
\end{equation}
This fixes $g$ by linearity.
It is easy to check that the equation \eqref{eq:circularity2} and the inequalities \eqref{eq:ab} hold.

In order to preserve the locality of a problem $E,S,\syndx$, it suffices to impose in addition that for each $\xi$ the set 
\begin{equation}
\bigcup_i \epsilon_i\cup \xi
\end{equation}
is connected.
This is not a constraint at all because, if it is not satisfied, it suffices to discard from each $\epsilon_i$ those elements that are not in the same connected component as the single element of $\xi$.
This is a consequence of the property \eqref{eq:conepsilon}.

\subsection{3D gauge color codes}

Error correction in 3D gauge color codes involves two steps.
The first one is the correction of measurement errors in the extraction of the gauge syndrome.
The second is the correction of errors based on the (stabilizer) syndrome.
Both problems have a similar nature in that they are described in terms of a charge group $\mathbf Z_2^3$, and both can be put in a form that satisfies the above requirements \eqref{eq:epsilon_i} while preserving locality.
There are, however, some interesting obstacles that deserve a detailed discussion.

\subsubsection*{Correction of bit-flips}

Take for example the problem $E,S,\syndx$ where the elements of $E$ are bit-flip errors and the elements of $S$ are the $Z$-type stabilizer generators.
Recall that stabilizer generators have a color $\kappa$ in the set 
\begin{equation}\label{eq:colors}
\sset{\text r,\text g,\text b,\text y}.
\end{equation}
Consider first a lattice without boundaries.
Then the $\mathbf Z_2^3$ charge is well-defined globally, \emph{i.e.} a subset of $S$ is a syndrome if and only if it has neutral charge (the sum of the labels of its elements is trivial).
A single bit-flip error $\xi$ in the bulk has a syndrome $\sigma=\sset{s_r,s_g,s_b,s_y}$, where $s_\kappa$ is a stabilizer at a $\kappa$-vertex~\cite{bombin:2007:branyons}.
The problem is that there is no way to partition $\sigma$ into subsets that have at most two elements and are syndromes of some error: subsets with one or two elements do not have have neutral charge.
This difficulty can be overcome by choosing a minimal set $C$ of charges that generate the group, \emph{e.g.} 
\begin{equation}\label{eq:C1}
C=\sset{\text{r}, \text g, \text b},
\end{equation}
and an alternative set of local stabilizer generators that have charge either neutral or in $C$, see~\cite{bombin:2013:structure} for the explicit procedure.
Given any elementary error $\xi$ (or in fact any error), if its syndrome consists of $n_0$ neutral elements and, for each $\kappa\in C$, $n_\kappa$ $\kappa$-elements, then, due to the neutrality of $\synd \xi$,
\begin{equation}\label{eq:neutral}
n_\kappa\equiv 0\mod 2.
\end{equation}
Thus $\sigma$ can always be partitioned by forming (i) singlets of neutral elements and (ii) pairs consisting of elements with the same charge. 
The corresponding $\epsilon_i$ of \eqref{eq:epsilon_i} can always be local.

In the presence of boundaries charge is only interesting as a local property.
In particular, in the bulk the charge group is $\mathbf Z_2^3$, but in the vicinity of a free $\kappa$-region, with $\kappa$ in the set \eqref{eq:colors}, the relevant local charge group is the quotient of $\mathbf Z_2^3$ with the $\mathbf Z_2$ subgroup generated by $\kappa$ (accordingly, in the vicinity of several free regions the quotient has to be taken instead with the group generated by the respective colors).
The absence of an interesting global notion of charge is no obstacle to apply the procedure of~\cite{bombin:2013:structure}, since it is possible to use the labels of vertices and the corresponding abstract $\mathbf Z_2^3$ group.
The resulting set of stabilizer generators have labels in $C\cup\sset 0$.
Either in the bulk or the boundary, the local charge of a stabilizer generator matches its label via the quotient construction, \emph{i.e.} a set $\sigma\subseteq S$ is \emph{locally neutral} (its labels add up to 0 in the local charge group) if and only if there exists a local operator with syndrome $\sigma$.

Since the construction is the same as before in the bulk, bit-flip errors $\xi$ in the bulk have the desired properties. 
This however might not be the case for bit-flip errors in the boundary.
A single bit-flip error in a y-region can have a syndrome of the form $\sigma=\sset{s_r,s_g,s_b}$, where $s_\kappa$ has color $\kappa$, because such a set is locally neutral.
As above, the problem is that there is no way to divide $\sigma$ into subsets with locally neutral charge and at most two elements.
This can be overcome by choosing a different set $C$ of charge group generators, namely
\begin{equation}\label{eq:C2}
C=\sset{\text{r}, \text g, \text x},\qquad \text x:=\text r+b.
\end{equation}
Given any elementary error $\xi$ (or in fact any local error), if its syndrome consists of $n_0$ neutral elements and, for each $\kappa\in C$, $n_\kappa$ $\kappa$-elements, then, due to the local neutrality of $\synd \xi$, either \eqref{eq:neutral} holds or instead
\begin{equation}
n_{\text r}\equiv n_{\text g}+1\equiv n_{\text x}+1\equiv 0\mod 2.
\end{equation}
In the latter case, $\sigma$ can always be pertitioned by forming (i) singlets of neutral elements, (ii) pairs consisting of elements with the same label and (iii) a pair of elements with labels g and x.
Ultimately the partition is possible simply because
\begin{equation}
\text y=\text g+\text x.
\end{equation}
In particular, there exist a set of generators for the group of locally neutral charge (namely,$\sset{\text y}$), that are the sum of less than three elements of $C$ (y is the sum of two of them).
The set of generators \eqref{eq:C2}, unlike \eqref{eq:C1}, is such that this is true for some generating set of all the locally neutral charge groups.
In particular, it is always possible to choose the generators from the set \eqref{eq:colors}, and (the case of r and g is trivial)
\begin{equation}
\text b=\text r+\text x.
\end{equation}
It is crucial that boundary conditions only include r-, g-, b- and y-regions.
Other boundary conditions might imply different and possibly incompatible constraints on the valid choices for $C$.

\subsubsection*{Correction of measurement errors}

The correction of measurement errors (in the extraction of the $Z$-type gauge syndrome) is entirely analogous. 
The elements of $E$ are links and the elements of $S$ are pairs $(v,x)$ with $v$ an internal vertex and $x$ one of the two generators of the $\mathbf Z_2\times\mathbf Z_2$ subgroup of charges that a given vertex can hold.
\emph{E.g.} a r-vertex can only hold charge 0, gb, by or gy, and any two nontrivial labels among these might be taken as generators.
In this case a valid choice of generators of the charge group $\mathbf Z_2^3$ is 
\begin{equation}\label{eq:Cgauge}
C=\sset{\text{rg}, \text {gb}, \text {by}}.
\end{equation}
Indeed, the generators of the groups of locally neutral charges can be chosen among the elements of
\begin{equation}
C\cup \sset{\text{rb, gy}},
\end{equation}
and
\begin{equation}
\text {rb}=\text {rg}+\text {gb},\qquad
\text {gy}=\text {gb}+\text {by}.
\end{equation}

It is worth mentioning a special case with an important application~\cite{bombin:2014:dimensional}.
When the original state is known to not have branching points, the correction of measurement errors is greatly simplified.
The reason is that each of the 6 flux types can be treated separately, \emph{i.e.} a gauge syndrome $\delta$ is decomposed in syndromes $\delta_{\kappa}$ with $\kappa$ a pair of colors. 
Each problem is described by a single $\mathbf Z_2$ charge, and in fact is directly isomorphic to that in section~\ref{sec:efficient_solution}.

\section{Conclusions and outlook}

Quantum-local operations, in particular single-shot error correction, could play a crucial role in fault-tolerant quantum computation.
The powerful features of 3D gauge color codes are a great example of this: with single-shot error correction it is a `straightforward' matter to perform all elementary operations in constant time (disregarding classical computation time).

At the most basic level, one can think of quantum error correcting codes simply as subspaces.
On top of that, explicit algorithms for error correction are often crucial: in the absence of efficient algorithms, a code might be unpractical or a threshold irrelevant.
Even this is a simplified picture when considering fault-tolerance, where measurement errors play a crucial role. 
This is exemplified by the techniques for error correction in local stabilizer codes~\cite{dennis:2002:tqm}, and now also by single-shot error correction.
It seems that rather than the code, the focus should be on the error correcting procedure.

In this regard, a family of codes might display \emph{locality} in at least three levels of strength.
The largest class is that of codes that admit a description in terms of local projectors.
Then comes the class of codes for which ideal quantum error correction is quantum-local.
This is \emph{different} from the previous, as exemplified by the ground subspaces of system with non-abelian anyons: it is not quite the same to be able to detect errors locally and to be able to obtain the error syndrome locally.
Finally there is the class of codes that admit single-shot error correction.

In order to understand the practical value of gauge color codes it will be important to develop decoding algorithms and to study the corresponding error thresholds, both in the context of single-shot error correction and of conventional error correction by repetition.
Remarkably, the present work has been quickly followed up by a first numerical simulation of the single-shot features of gauge color codes~\cite{brown:2015:fault}.
It is natural to conjecture that, compared to other 3D codes, gauge color codes may exhibit large error thresholds in the conventional fault-tolerant scenario if gauge syndromes are used.

A fundamental open problem is whether a self-correcting topological phase might exist in 3D.
In this regard, one can conjecture that a Hamiltonian system connected to 3D gauge color codes could give rise to a self-correcting phase~\cite{bombin:2013:gauge}.
The non-local encoding of quantum information in 3D gauge color codes is compatible with the confinement of the two types of particles ($X$- and $Z$-type) that will destroy it if free to move around.
The problem is that it is not clear whether these two confinements can be compatible, even if for example alternated in time.

\begin{acknowledgments} 

I received support from the MINECO grant FIS2009-10061, the CAM grant QUITEMAD+, the Sapere Aude grant of the Danish Council for Independent Research, the ERC Starting Grant QMULT and the CHIST-ERA project CQC.
Work at Perimeter Institute is supported by Industry Canada and Ontario MRI.

\end{acknowledgments}

\bibliography{refs,comments}

\end{document}